\documentclass[conference]{IEEEtran}
\IEEEoverridecommandlockouts
\usepackage{cite}
\usepackage{amsmath,amssymb,amsfonts}
\usepackage{graphicx}
\usepackage{textcomp}
\usepackage{xcolor}

\usepackage{enumitem}
\usepackage{mathcommands}
\usepackage{subcaption}
\usepackage{diagbox}
\usepackage{booktabs}
\usepackage[linesnumbered,ruled,vlined,noend]{algorithm2e}

\usepackage{amsthm}

\newtheorem{theorem}{Theorem}
\newtheorem{corollary}{Corollary}
\newtheorem{lemma}{Lemma}

\newtheorem{definition}{Definition}
\newtheorem{assumption}{Assumption}

\newtheorem{example}{Example}
\newtheorem{problem}{Problem}

\def\BibTeX{{\rm B\kern-.05em{\sc i\kern-.025em b}\kern-.08em
    T\kern-.1667em\lower.7ex\hbox{E}\kern-.125emX}}
\begin{document}

\title{A Scalable Game-theoretic Approach to Urban Evacuation Routing and Scheduling%\\
% {\footnotesize \textsuperscript{*}Note: Sub-titles are not captured for https://ieeexplore.ieee.org  and
% should not be used}
% \thanks{Identify applicable funding agency here. If none, delete this.}
}

\author{
    \IEEEauthorblockN{
        Kazi Ashik Islam,
        Da Qi Chen,
        Madhav Marathe,
        Henning Mortveit,
        Samarth Swarup,
        Anil Vullikanti
    }
    \IEEEauthorblockA{
        \textit{Biocomplexity Institute, University of Virginia}\\
        % City, Country \\
        ki5hd@virginia.edu, daqic@alumni.cmu.edu, \{marathe, henning.mortveit, swarup, vsakumar\}@virginia.edu
    }
}

\newcommand{\minsumpathproblem}{\textsc{mspp}}
\newcommand{\bestresponseproblem}{\textsc{brp}}
\newcommand{\avgcostproblem}{\textsc{a-dcfp}}
\newcommand{\completiontimeproblem}{\textsc{ct-dcfp}}
\newcommand{\feasibilityproblem}{\textsc{f-dcfp}}
\newcommand{\equilibriumproblem}{\textsc{e-dcfp}}
\newcommand{\evacuationgame}{\textsc{egres}}
\newcommand{\confluentevacuationgame}{\textsc{con-egres}}
\newcommand{\bestresponsealgo}{\textsc{brsa}}
\newcommand{\seqactionalgo}{\textsc{saa}}
\newcommand{\confluentconstr}{\textsc{crc}}

\maketitle

\begin{abstract}
Evacuation planning is an essential part of disaster management where the goal is to relocate people under imminent danger to safety. However, finding jointly optimal evacuation routes and a schedule that minimizes the average evacuation time or evacuation completion time, is a computationally hard problem. As a result, large-scale evacuation routing and scheduling continues to be a challenge.
% Although government authorities often prescribe routes and a schedule, evacuees generally behave as self-interested agents and may choose their action 
% (i.e., when to leave and which route to take) 
% according to their own selfish interests. It is crucial to understand the degree of inefficiency this can cause to the evacuation process. Existing research has mainly focused on selfish routing, i.e., they consider route selection as the only strategic action. 
In this paper, we present a game-theoretic approach to tackle this problem. We start by formulating a strategic routing and scheduling game, named the \textit{Evacuation Game: Routing and Scheduling} (\evacuationgame{}), where players choose their route and time of departure.
% We also use \textit{dynamic flows} to model the time-varying traffic on roads during evacuation. 
We show that: (i) every instance of \evacuationgame{} has at least one pure strategy Nash equilibrium, and (ii) an optimal outcome in an instance will always be an equilibrium in that instance. We then provide bounds on how bad an equilibrium can be compared to an optimal outcome. Additionally, we present a polynomial-time algorithm, the \textit{Sequential Action Algorithm} (\seqactionalgo{}), for finding equilibria in a given instance under a special condition.
%, we call \textit{Confluent-Evacuation Game: Routing and Scheduling, \confluentevacuationgame{}}. 
We use Virginia Beach City in Virginia, and Harris County in Houston, Texas as study areas and construct two \evacuationgame{} instances. 
% We apply \seqactionalgo{} to find equilibrium states in this instance. 
Our results show that, by utilizing \seqactionalgo{}, we can efficiently find equilibria in these instances that have social objective close to the optimal value.
\end{abstract}


\section{Introduction}
Evacuation plans are designed to ensure the safety of people living in areas that are prone to natural or man-made disasters such as hurricanes, tsunamis, wildfires, and toxic chemical spills. Large-scale evacuations have been carried out during the past hurricane seasons in Florida, Texas, Louisiana, and Mississippi. For instance, about $2.5$ million people were ordered to evacuate from Florida before the landfall of Hurricane Ian (2022)~\cite{fox_weather_ian_evacuation_2022, abc_news_ian_evacuation_2022}.
Other examples of hurricanes when such large scale evacuations were carried out include Helene \& Milton (2024), Idalia (2023), Ida (2021), Laura (2020), Irma \& Harvey (2017), Ike \& Gustav (2008), Katrina \& Rita (2005).
It is, therefore, necessary to have effective and efficient evacuation plans in place to ensure the sustainability of cities and communities. Such a plan needs to determine two major components: $(i)$ Evacuation Routes, recommending the evacuees which route to take, and $(ii)$ Evacuation Schedule, recommending the evacuees when to leave.  

In this paper, we approach the problem of finding jointly optimal routes and schedule that minimizes the sum (or average) of the evacuation time of all evacuees. These routes and schedule can then be recommended to the evacuees. Existing research works have shown that this joint optimization problem is not only {\sf NP-hard} but also hard to approximate \cite{Golin2017NonapproximabilityAP,islam2023optimal}. Majority of the solution methods proposed so far (see Section \ref{sec:related_works} for details) utilize 
% Mixed Integer Program (MIP) formulations of the problem. These formulations make use of a 
\textit{time-expanded graphs} to capture the time-varying traffic load on the roads. 
A time-expanded graph contains a copy of each node and edge in the road network at each timestep within a time horizon.
With the extensive road networks of current cities and the large population of urban residents (which necessitates a prolonged evacuation time period), the size of time-expanded graph grows very quickly. As a result, the existing solution methods do not scale to large city or county level problems. Thus, designing a scalable solution approach that finds good evacuation routes and schedule for real-world problem instances, remains an open problem.

In this paper, we present a game-theoretic approach to solving the above-mentioned problem. 
As our \textbf{first} contribution, we formulate a strategic routing and scheduling game, named the \textit{Evacuation Game: Routing and Scheduling} (\evacuationgame{}), where players choose their routes and departure schedule. Each player's goal is to minimize their own evacuation time. However, they cannot exceed the capacity of the roads in the road-network. By formulating \evacuationgame{}, we are able to study it from a game-theoretic perspective. 
As our \textbf{second} contribution, we show that every instance of \evacuationgame{} has at-least one pure strategy Nash equilibrium. We also show that an optimal outcome (i.e., a set of routes and departure schedule chosen by the players that minimizes the sum of the evacuation time of all players) in an \evacuationgame{} instance will always be an equilibrium in that instance. We can, therefore, consider the problem of finding equilibria in a given \evacuationgame{} instance. 
As our \textbf{third} contribution, we provide bounds on how bad an equilibrium outcome can be compared to an optimal outcome. We show that, by carefully choosing the road-network structure and road attributes, it is possible to construct \evacuationgame{} instances where an equilibrium outcome is worse than an optimal outcome.
As our \textbf{fourth} contribution, we present a \textit{polynomial-time algorithm} (\textit{Sequential Action Algorithm}, \seqactionalgo{}) for finding equilibria in a given \evacuationgame{} instance, under a special condition (see Section \ref{sec:eq_calc} for details). \seqactionalgo{} allows us to quickly find equilibria in large \evacuationgame{} instances. 
\textbf{Finally}, we use Virginia Beach City in Virginia, and Harris County in Houston, Texas as our study area. By using real-world road-network data and a synthetic population data, we construct an \evacuationgame{} instance for each of these two areas. Through our experiments we show that: $(i)$ \seqactionalgo{} finds equilibria in our game instances that have social objective (i.e. sum of the evacuation time of all evacuees) close to the optimal social objective, and $(ii)$ \seqactionalgo{} outperforms time-expanded graph based methods in terms of scalability. %We, therefore, propose \seqactionalgo{} as a scalable solution method for finding and recommending evacuation routes and schedule.
\section{Related Works}
\label{sec:related_works}
The evacuation routing and/or scheduling problem have been approached in several ways in the literature. A review of existing works can be found in the survey paper~\cite{Bayram2016}. Hamacher and Tjandra~\cite{hamacher2002} formulated the problem as a dynamic network flow optimization problem. They introduced the idea of time-expanded graphs that contain a copy of each node and each edge in the road network at each timestep within a time horizon. The authors then solved the problem using mathematical optimization methods. However, their proposed method had very high computational cost. This led the researchers to design and propose several heuristic methods~\cite{Lu:2005:CCR:2156226.2156249,Kim:2007:ERP:1341012.1341039,Shahabi2014,islam2020simulation}. The drawback of these methods is that they only consider the routing problem, i.e., they either do not consider the scheduling problem, or make simplifying assumptions such as evacuees leave at a constant rate. 

Even and Pillac~\emph{et al.}~\cite{even2015convergent}, 
% Pillac and Cebrian~\emph{et al.}~\cite{Pillac2016136}, 
Pillac and Van Hentenryck~\emph{et al.}~\cite{Pillac2015}, Romanski and Van Hentenryck~ \cite{romanski2016benders}, Hasan and Van Hentenryck~\cite{hafiz2021large}, have considered the joint optimization problem of routing and scheduling. They utilize time-expanded graphs to consider time-varying traffic load on the roads, and then formulate the optimization problem as Mixed Integer Programs (MIP). However, the size of time-expanded graphs causes a bottleneck in these approaches. Evacuation scenarios considered within these works have small road-networks ($\sim300$ nodes, $\sim600$ edges) and small number of evacuating vehicles ($\sim100,000$). Moreover, to keep the size of the time-expanded graph manageable, each timestep is defined to have 5 minutes; which is quite coarse-scale considering the high-speed roads available in modern cities. Islam and Chen~\emph{et al.}~\cite{islam2023optimal,Islam2022fairness} proposed heuristic approaches to address these issues.
% and applied different speedup techniques to reduce the solve time of their MIP. These techniques include decomposition techniques \cite{benders1962partitioning, Magnanti1981} that separates the route selection and scheduling process, and heuristic search techniques. 
In contrast, we present a game-theoretic approach that does not rely on time-expanded graphs or MIP formulations. We first formulate a strategic routing and scheduling game and show that every instance of this game has a pure strategy Nash equilibrium. We then present a \textit{polynomial-time} algorithm for finding equilibria. 
% under a special constraint, namely the \textit{Confluent Route Constraint} (\confluentconstr{}). 
As a result, we are able to consider larger evacuation scenarios (road networks with $\sim3,500$ nodes and $\sim7,000$ edges; $\sim1.5$ million evacuating vehicles), and fine-scale timesteps (e.g. $30$ second timesteps).

Existing works on evacuation planning has explored the use of confluent routes~\cite{even2015convergent,romanski2016benders,hafiz2021large, Islam2022fairness,islam2023optimal}. In emergency evacuation scenarios, evacuees tend to follow a leader-and-follower pattern~\cite{DING2020103189}, which leads them to take the same outgoing road when they reach an intersection. As a result, when multiple evacuation routes meet at a node, their subsequent parts become identical. Such routes are known as confluent routes (see example in Figure \ref{fig:confluent_routes}). Confluent routes do not contain any forks. This reduces driver hesitancy, which has been found to be a significant source of delays~\cite{townsend2006federal, wolshon2006louisiana}. Moreover, enforcing confluent routes is easy in practice; after blocking the necessary roads, minimal vehicle guidance is needed~\cite{romanski2016benders}. Due to these practical advantages, Even and Pillac~\emph{et al.}~\cite{even2015convergent}, Romanski and Van Hentenryck~\cite{romanski2016benders}, Hasan and Van Hentenryck~\cite{hafiz2021large}, and Islam and Chen~\emph{et al.}~\cite{Islam2022fairness,islam2023optimal} have considered confluent evacuation planning, where they aim to find an optimal set of confluent routes and schedule. Islam and Chen~\emph{et al.}~\cite{islam2023optimal} proved that this problem is {\sf NP-hard} and also hard to approximate. Because of the practical benefits of using confluent routes, we compute equilibria under the Confluent Route Constraint in our proposed method. 
% This also allows us to design a polynomial-time algorithm that is highly scalable.
\section{Problem Formulation}
In this section, we introduce some preliminary terms and then present the formulation of our game.

\begin{definition}
\label{def:evac_network}
A \textbf{road network} is a directed graph $\cG = (\cV, \cA)$ where every edge $e \in \cA$ has ($i$) a capacity $c_e\in\NN$ representing the number of vehicles that can enter the edge at each timestep and ($ii$) a time $\tau_e\in\NN$ representing the time it takes to traverse the edge. An \textbf{evacuation network} is a road network that specifies $\cE, \cS, \cT \subset \cV$, representing a set of source, safe and transit nodes respectively. For each source node $k\in \cE$, $W(k)$ represents the set of evacuees at source node $k$. We assume that $\cE$ and $\cS$ are disjoint (i.e. $\cE \cap \cS = \emptyset$). We also assume that there is a path from each source node $k \in \cE$ to at least one of the safe nodes.
\end{definition}

\begin{definition}
\label{def:evac_game}
Given an evacuation network $\cG$ and an upper bound on evacuation time $T_{max}$, we define the \textbf{Evacuation Game: Routing and Scheduling (\evacuationgame{})} as follows:
\begin{itemize}[leftmargin=*]
    \item \textbf{Players}: We have $N = |\cE|$ players denoted by the set $[N] = \{1, 2, ..., N\}$ where player $i$ corresponds to the source $src_i \in \cE$.
    \item \textbf{Actions}: Player $i$ can take actions $a_i \in A_i$ where $A_i = R_i \times DT_i$. 
    \begin{itemize}[leftmargin=*]
        \item $R_i$ is the set of all possible simple paths (i.e. paths with no cycles) from source $src_i$ to any safe node in $\cS$. 
        \item $DT_i$ is the set of all possible departure time schedules of the evacuees at source $src_i$. We represent a departure time schedule `$dt$' as follows:
        \begin{align*}
            dt &= \{(t, \theta_t)\ |\ t \in [0, T_{max}-1],\\
            & \theta_t = \text{Number of evacuees departing at timestep $t$}\}
        \end{align*}
    \end{itemize}
    
    \item \textbf{Outcome}: An outcome is the action profile $a = (a_1, a_2, ..., a_N)$ where, each player $i$ has chosen a particular action $a_i \in A_i$. In other words, each player chooses a path (to a safe node) and a departure schedule, in outcome $a$. We say: $a \in A, \text{where } A = A_1 \times A_2 \times ... \times A_N$.
    
    \item \textbf{Utility}: Each player $i$ has a utility function $u_i$. With outcome $a$, utility of player $i$ is denoted by $u_i(a)$, where:
    \begin{align}
        \label{eq:utility}
        & u_i(a) = \begin{cases}
            -\sum_{l \in W(i)}t_l & \text{if Case $1$}\\
            -M_1 & \text{Otherwise}
        \end{cases}\\
        & \text{where, }t_l = t_l^d + \sum_{e \in route_i} \tau_e \nonumber
    \end{align}
    Here, $t_l^d$ denotes the departure time of evacuee $l$ from source $src_i$, $\tau_e$ denotes travel time on edge $e$, $route_i$ denotes the set of edges in player $i$'s route, $t_l$ denotes the evacuation time of evacuee $l$, $W(i)$ denotes the set of evacuees under player $i$, and $M_1$ is a very large positive number.
    Case $1$ occurs when the following are true:
    \begin{enumerate}
        \item No edge on player $i$'s route, at any point in time, exceeds its capacity.
        % \item $\forall j \in [N], j \neq i$, $a_i$ and $a_j$ together induces a dynamic confluent flow.
        \item All evacuees under player $i$ reach a safe node within time $T_{max}$.
    \end{enumerate} 
    % Note that, by this definition, $u_i(a) \geq -M_1, \forall i \in [N], a \in A$.
\end{itemize}
\end{definition}

An example instance of \evacuationgame{} is provided in Example \ref{example:epg_instance}. In \evacuationgame{}, we make the following assumption:
\begin{assumption}
\label{assumption:completion_time_upper_bound}
Given an instance $\cI$ of \evacuationgame{} with the evacuation network $\cG(\cV, \cA)$, we assume that $T_{max}=n\tau + M -1$. Here, $n$ is the number of nodes in $\cG$, $M$ is total number of evacuees, and $\tau = \sum_{e \in \cA} \tau_e$.
\end{assumption}

The reason for Assumption \ref{assumption:completion_time_upper_bound} is the following lemma:
\begin{lemma}
\label{lemma:cmpletion_time_upper_bound}
Given an evacuation network $\cG(\cV, \cA)$, there exists an outcome where all evacuees arrive at some safe node within time $(n\tau + M - 1)$ without any edge capacity violation. Here, $n$ is the number of nodes in $\cG$, $M$ is total number of evacuees, and $\tau = \sum_{e \in \cA} \tau_e$.
\end{lemma}

% From the definition of \evacuationgame{} (Definition \ref{def:evac_game}), we have:
% \begin{observation}
% \label{observation:flow_game_eq}
% Let $\cI$ be an instance of \evacuationgame{} with the evacuation network $\cG(\cV, \cA)$, and evacuation time upper bound $T_{max}$. Then, in an outcome `$a$' of $\cI$, all players get a utility greater than $-M_1$ if and only if there is no edge capacity violation in $\cG$ with outcome `$a$' and all evacuees reach a safe node within time $T_{max}$.  
% \end{observation}

Proof of Lemma \ref{lemma:cmpletion_time_upper_bound} is provided in the Appendix.
Let $a_{-i}$ denote actions of all players excluding player $i$ in outcome $a$, i.e. $a_{-i} = (a_1, ..., a_{i-1}, a_{i+1}, ..., a_n)$. Then: 
% together with the definition of our game, we can refer to the following definitions from Game Theory:

% \begin{definition}
% Let $a_{-i}$ be fixed. Then, $a_i^*$ is called a best response to $a_{-i}$ if:
% $$u_i(a_i^*, a_{-i}) \geq u_i(a_i, a_{-i}), \forall a_i \in A_i$$
% \end{definition}

% \begin{definition}
% An action $a_i$ is a \textit{dominant strategy} for player $i$ if $a_i$ is better than any other action $a_i^\prime \in A_i$, regardless what actions other players take. Formally,
% $$u_i(a_i, a_{-i}) \geq u_i(a_i^\prime, a_{-i}), \forall a_i^\prime \neq a_i \text{ and } \forall a_{-i}$$
% \end{definition}

\begin{definition}
\label{def:pne}
An outcome $a^*$ is an \textbf{equilibrium} if no player has incentive to deviate unilaterally. Formally,
$$u_i(a_i^*,a_{-i}^*) \geq u_i(a_i, a_{-i}^*), \forall a_i \in A_i \text{ and } \forall i \in [N]$$
\end{definition}

Note that an equilibrium here is a pure strategy Nash equilibrium. 
% We also define a \textit{socially optimal outcome} as follows:

\begin{definition}
\label{def:social_opt}
An outcome $a^*$ is \textbf{socially optimal} if the sum of the utility of all players is maximum over all possible outcomes. Formally, $$a^* = \argmax_{a \in A} \sum_{i \in N} u_i(a)$$
% $$\sum_{i \in [N]} u_i(a^*) \geq \sum_{i \in [N]} u_i(a), \forall a \in A$$
\end{definition}

To quantify the quality of equilibria, \textit{Price of Anarchy (PoA)} and \textit{Price of Stability (PoS)} is defined as follows:
\begin{definition}
\label{def:poa}
Given an instance $\gamma$ of \evacuationgame{}, let $EQ(\gamma)$ denote the set of equilibrium outcomes in $\gamma$. Let, $U(a) = \sum_{i \in [N]}u_i(a)$. 
% Also, let $a^*$ denote a socially optimal outcome of $\gamma$, and therefore, $U(a^*)$ denote the optimal social objective value in $\gamma$. 
Then, \textbf{Price of Anarchy} for the instance $\gamma$ is: 
\begin{align}
    \rho(\gamma) = \frac{\min_{a \in EQ(\gamma)}U(a)}{\max_{a \in A}U(a)}
\end{align}

On the other hand, \textbf{Price of Stability} for the instance $\gamma$ is: 
\begin{align}
    \sigma(\gamma) = \frac{\max_{a \in EQ(\gamma)}U(a)}{\max_{a \in A}U(a)}
\end{align}
% Let $\Gamma$ be a set of instances of \evacuationgame{}. Then, price of anarchy of $\Gamma$ is:
% \begin{align}
%     \rho(\Gamma) = \sup_{\gamma \in \Gamma} \rho(\gamma)
% \end{align}
\end{definition}

\begin{example}
\label{example:epg_instance}
\begin{figure}[!b]
    \centering
    \begin{subfigure}[b]{0.45\textwidth}
    \centering
    \includegraphics[width=0.75\textwidth]{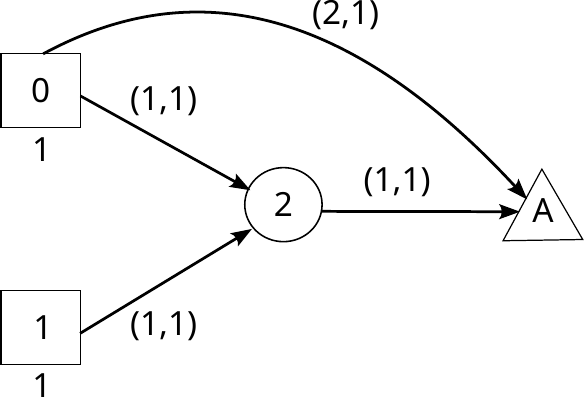}
    \caption{Evacuation network of the example \evacuationgame{} instance. Edges are labeled with travel time and flow capacity, respectively. Source, safe, and transit nodes are denoted by squares, triangles, and circles, respectively. Source nodes are labeled with number of evacuees (the number under the square).}
    \label{fig:example_3_net}
    \end{subfigure}
    % \hfill
    \par\bigskip
    \begin{subfigure}[b]{0.5\textwidth}
        \centering
        \small
		\begin{tabular}{p{1.25cm}p{1.75cm}p{1.75cm}p{1.75cm}}
        \toprule
        \diagbox{$P_0$}{$P_1$} & $[1,2,A],$\newline$[(0,1)]$ & $[1,2,A],$\newline$[(1,1)]$ & $[1,2,A],$\newline$[(2,1)]$\\
        \midrule
        $[0,2,A],$\newline$[(0,1)]$ & \textcolor{orange}{$(-M_1,-M_1)$} & \textcolor{blue}{$(-2,-3)$} & $(-2, -4)$\\
        \midrule
        $[0,2,A],$\newline$[(1,1)]$ & $(-3,-2)$ & \textcolor{orange}{$(-M_1,-M_1)$} & $(-3, -4)$\\
        \midrule
        $[0,2,A],$\newline$[(2,1)]$ & $(-4,-2)$ & $(-4,-3)$ & \textcolor{orange}{$(-M_1,-M_1)$}\\
        \midrule
        $[0,A],$\newline$[(0,1)]$ & \textcolor{green}{$(-2,-2)$} & $(-2,-3)$ & $(-2, -4)$\\
        \midrule
        $[0,A],$\newline$[(1,1)]$ & $(-3,-2)$ & $(-3,-3)$ & $(-3, -4)$\\
        \midrule
        $[0,A],$\newline$[(2,1)]$ & $(-4,-2)$ & $(-4,-3)$ & $(-4, -4)$\\
        \bottomrule
        \end{tabular}
        \caption{Utility table for the example \evacuationgame{} instance.}
        \label{tab:example_3_utility}
    \end{subfigure}
    \caption{\evacuationgame{} example instance with $T_{max}=4$.}
    % \vspace{-3mm}
\end{figure}

Figure \ref{fig:example_3_net} (evacuation network) and \ref{tab:example_3_utility} (utility table) show an example instance of \evacuationgame{}. Here, we have two players: $P_0$ and $P_1$, corresponding to the two source nodes $0$ and $1$ in Figure \ref{fig:example_3_net}. There is only one safe node ($A$) in this instance. Each player has $1$ evacuee under her.
Figure \ref{tab:example_3_utility} shows the possible actions of the two players. Actions of $P_0$ are shown in rows, and actions of $P_1$ in columns. For instance, the action $[0,2,A], [(0,1)]$ of $P_0$ means: $P_0$ has chosen the route $[0,2,A]$, and one evacuee leaves from node $0$ at timestep $0$.
Each cell of the utility table in Figure \ref{tab:example_3_utility} corresponds to an outcome.
The first and second value in each cell are the utility values of player $P_0, P_1$ respectively. These values are calculated according to Equation \ref{eq:utility}.

In the orange-colored outcomes (Figure \ref{tab:example_3_utility}), utility of both players is $-M_1$. This is because, in these outcomes capacity violation occurs on edge $(2,A)$ at some timestep. For instance, when $P_0$ and $P_1$ choose the actions $[0,2,A], [(0,1)]$ and $[1,2,A], [(0,1)]$ respectively (orange outcome in the top-left cell), $2$ evacuees enter the edge $(2, A)$ at timestep $1$. This causes a capacity violation as capacity of edge $(2, A)$ is $1$.

The green-colored outcome is socially optimal because it has the highest total utility of $-4$. It is also the only socially optimal outcome in this example.

The blue-colored outcome is an equilibrium because none of the players have incentive to unilaterally deviate from it. In this outcome, the total utility of the players is $-5$, which is sub-optimal. Note that the green-colored outcome is also an equilibrium. The \textit{Price of Anarchy} of this particular instance is $5/4$, whereas \textit{Price of Stability} is $1$.

The upper-bound on evacuation completion time in this example is $T_{max}=4$. 
% This means, if all the evacuees under a player do not reach safe node $A$ within timestep $4$, then that player will receive an utility of $-M_2$ (not shown in Figure \ref{tab:example_3_utility}). Evacuees can start leaving from their source node at timestep $0$.
If the evacuee under $P_0$ leaves on or after timestep $3$, then the evacuee will not reach safe node $A$ within timestep $T_{max}$. In that case, the utility of $P_0$ will be $-M_1$. The same is true for $P_1$. Such outcomes are not shown in Figure \ref{tab:example_3_utility}.
\end{example}

Finding the socially optimal outcome in a given \evacuationgame{} instance is {\sf NP-hard} and also hard to approximate \cite{Golin2017NonapproximabilityAP,islam2023optimal}. Therefore, we focus on the following questions/problems:
\begin{problem}
\label{prob:equilibrium_existence}
Does every instance of \evacuationgame{} have a pure strategy Nash equilibrium?
\end{problem}
\begin{problem}
\label{prob:poa_bound}
Can we find bounds on the Price of Anarchy and Price of Stability of \evacuationgame{} instances?
\end{problem}
\begin{problem}
\label{prob:equilibrium_calc}
Given an instance of \evacuationgame{}, can we efficiently find an equilibrium in this instance?
\end{problem}

% \begin{problem}
% \label{prob:poa_bound}
% Can we find bounds on the Price of Anarchy of \evacuationgame{} over all possible instances?
% \end{problem}

% In both problems, we only consider equilibria where every player gets a utility greater than $-M_2$.

\section{Theoretical Analysis of \evacuationgame{}}
In this section, we consider the problems \ref{prob:equilibrium_existence}--\ref{prob:equilibrium_calc}. 
% In all three problems, we only consider equilibria where every player gets a utility greater than $-M_1$.

\subsection{Existence of Pure Strategy Nash Equilibrium}
To solve Problem \ref{prob:equilibrium_existence}, we prove the following theorem:
\begin{theorem}
\label{theorem:eq_existence}
Every instance of \evacuationgame{} has at least one pure strategy Nash equilibrium where all players get a utility greater than $-M_1$.
\end{theorem}
\begin{proof}
For any instance $\gamma\left([N], A, (u_i)_{i \in [N]}\right)$ of \evacuationgame{}, we prove that there exists an outcome: 
\begin{align}
    a^* = \argmax_{a \in A}\sum_{i \in [N]}u_i(a) \label{eq:socially_opt_outcome}
\end{align}
where $u_i(a^*) > -M_1, \forall i \in [N]$, and $a^*$ is an equilibrium.

Let, $U(a) = \sum_{i \in [N]}u_i(a)$. Then, $a^* = \argmax_{a \in A}U(a)$.

Since $\gamma$ has a finite number of players and each player has a finite number of actions, there are only a finite number of outcomes. Therefore, one of these outcomes, we denote it by $a^*$, is a global maximum of the function $U$.

Now, due to Assumption \ref{assumption:completion_time_upper_bound} and Lemma \ref{lemma:cmpletion_time_upper_bound}, $\gamma$ has at-least one outcome $a^\prime$ where $u_i(a^\prime) > -M_1, \forall i \in [N]$. Then, it is also true that $u_i(a^*) > -M_1, \forall i \in [N]$. Because otherwise $\sum_{i \in [N]}u_i(a^\prime) > \sum_{i \in [N]}u_i(a^*)$ and $a^*$ would not be the global maximum of $U$.

Finally, we prove that $a^*$ is an equilibrium. We prove this by contradiction. Let's assume that $a^*$ is not an equilibrium. Then, there exists a player $i$ who can change her action from $a^*_i$ to $\hat{a}_i$ such that:
\begin{align}
    u_i(\hat{a}_i, a^*_{-i}) > u_i(a^*_i, a^*_{-i}) > -M_1 \label{eq:alt_better_action}
\end{align}

Let, $\hat{a} = (\hat{a}_i, a^*_{-i})$.

Since $u_i(\hat{a}) > -M_1$, the switching of player $i$ from action $a^*_i$ to $\hat{a}_i$, cannot cause any edge capacity violation. Moreover, as travel time on each edge is constant, the change in player $i$'s action cannot change the evacuation time of evacuees under other players. Therefore, the change in player $i$'s action $(i)$ cannot cause any of the other players utility to become $-M_1$, and also $(ii)$ cannot decrease the utility of the other players. This means: $u_j(\hat{a}) \geq u_j(a^*), \forall j \in [N], j \neq i$. Combining this with inequality \ref{eq:alt_better_action}:

\begin{align}
    & u_i(\hat{a}) + \sum_{j \in [N], j \neq i}u_j(\hat{a}) \quad  > \quad u_i(a^*) + \sum_{j \in [N], j \neq i}u_j(a^*) \nonumber\\
    & \implies U(\hat{a}) > U(a^*) \nonumber
\end{align}

This means $a^*$ is not a global maximum of the function $U$, which is a contradiction.
\end{proof}

Since, at-least one equilibrium outcome exists in every instance of  \evacuationgame{}, we can now consider the questions: $(i)$ how good or bad an equilibrium can be compared to an optimal outcome, and $(ii)$ how do we compute equilibria in a given instance. From Theorem \ref{theorem:eq_existence}, we can deduce the following results directly:
\begin{corollary}
\label{corollary:opt_eq} Given any instance $\cI$ of \evacuationgame{}, an optimal outcome in $\cI$ will also be an equilibrium outcome in $\cI$. In other words, Price of Stability (PoS) of every \evacuationgame{} instance is $1$.
\end{corollary}
\begin{proof}[Proof sketch]
Notice that, $a^* = \argmax_{a \in A} \sum_{i \in N}u_i(a)$, in the proof of Theorem \ref{theorem:eq_existence}, is an optimal outcome. 
\end{proof}

\begin{corollary}
\label{corollary:potential_game} \evacuationgame{} is a potential game.
\end{corollary}
\begin{proof}[Proof sketch]
We show that $U(a)=\sum_{i \in N}u_i(a)$ is a generalized ordinal potential \cite{MONDERER1996124} for \evacuationgame{}. 
\end{proof}

Corollary \ref{corollary:opt_eq} indicates that an equilibrium can be as good as an optimal outcome. Corollary \ref{corollary:potential_game}, on the other hand, has important implications about the computation of equilibria. We discuss this further in Section \ref{sec:eq_calc}.

\subsection{Bound on Price of Anarchy}
\begin{figure}[t]
    \begin{subfigure}[b]{\linewidth}
        \centering
        \includegraphics[width=0.75\textwidth]{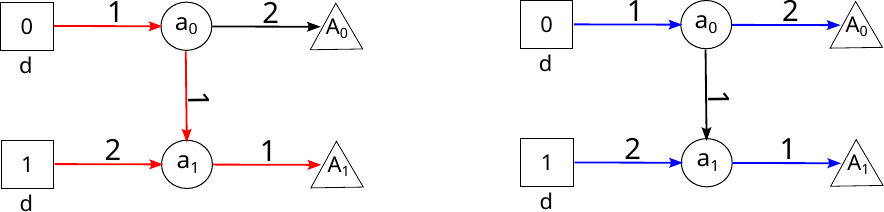}
        \caption{Level 2.}
        \label{fig:poa_unsplt_const_level_2}
    \end{subfigure}
    \par\bigskip
    \begin{subfigure}[b]{\linewidth}
        \centering
        \includegraphics[width=0.75\textwidth]{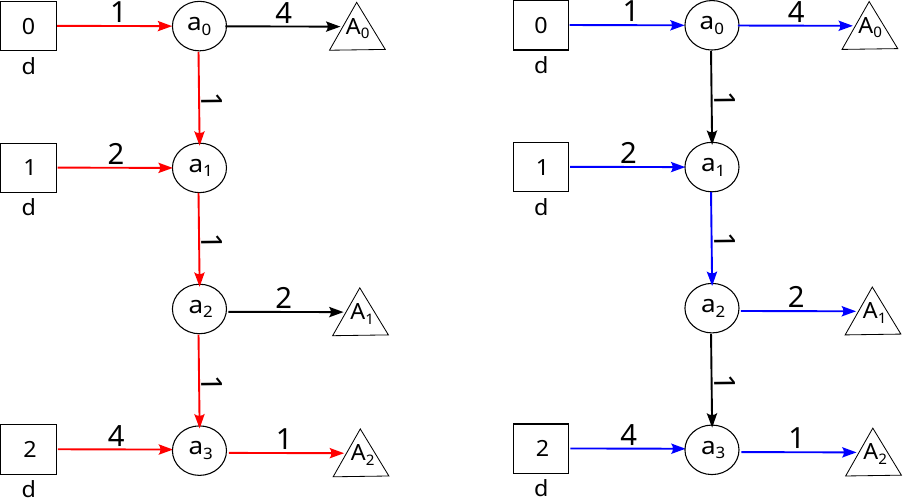}
        \caption{Level 3.}
        \label{fig:poa_unsplt_const_level_3}
    \end{subfigure}
    \par\bigskip
    \begin{subfigure}[b]{\linewidth}
        \centering
        \includegraphics[width=0.75\textwidth]{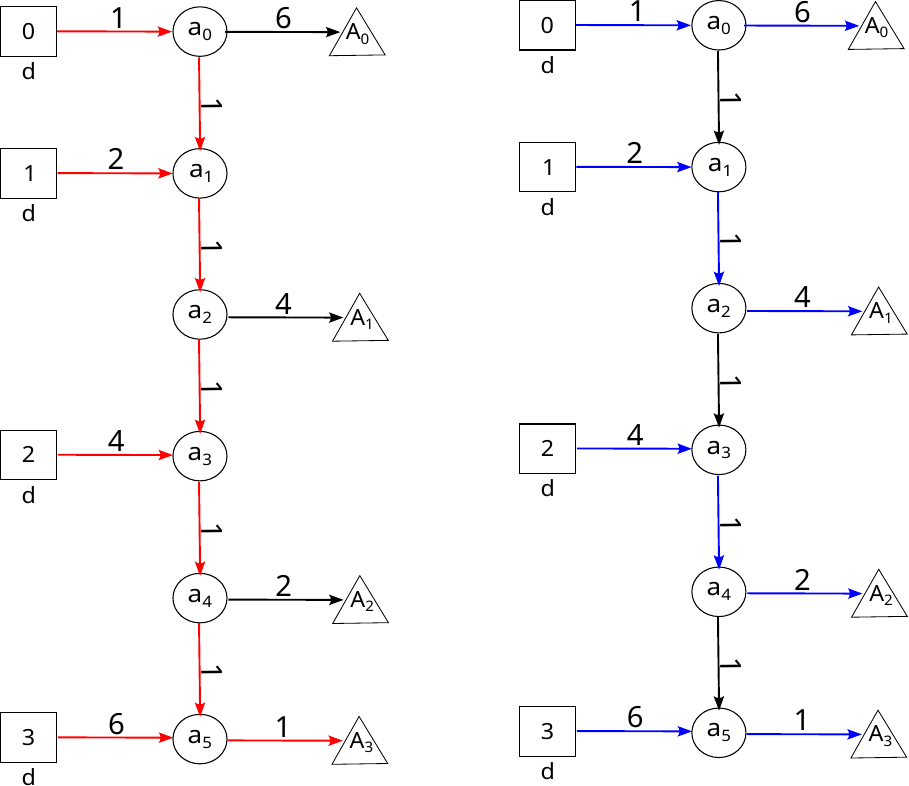}
        \caption{Level 4.}
        \label{fig:poa_unsplt_const_level_4}
    \end{subfigure}
    \caption{\evacuationgame{} instance where price of anarchy is $O(n)$. The edges are labelled with their travel time. All of the edges have capacity $1$. Each source node has $d$ number of evacuees.}
    \label{fig:poa_unsplt_const_lb}
\end{figure}
% In this section, we provide a theoretical bound on the price of anarchy of \evacuationgame{}. 
% over all possible instances of \evacuationgame{}. 
% Let, $\Gamma_{all}$ denote the set of all possible instances of \evacuationgame{}. 
% Then, we show that: the price of anarchy over all possible instances of \evacuationgame{}, i.e., $\rho(\Gamma_{all})$ is $\Theta(\tau + M)$. Here, $\tau = \sum_{e \in \cA}\tau_e$, and $M$ is the total number of evacuees.
In Definition \ref{def:poa}, we have defined $U(a) = \sum_{i \in [N]}u_i(a)$, i.e., the sum of the utility of all players for outcome $a$. Let $C(a) = -U(a)$, where $C(a)$ denotes a `cost' associated with outcome $a$. Then, for a given instance $\gamma$ of \evacuationgame{}:
\begin{align}
    \rho(\gamma) &= \frac{\min_{a \in EQ(\gamma)}U(a)}{\max_{a \in A}U(a)} = \frac{\min_{a \in EQ(\gamma)}-C(a)}{\max_{a \in A}-C(a)} \nonumber\\
    &= \frac{\max_{a \in EQ(\gamma)}C(a)}{\min_{a \in A}C(a)} \label{eq:poa_by_cost}
\end{align}

If in outcome $a$, all the players get a utility greater than $-M_1$, then $C(a)$ denotes the sum of the evacuation time of all evacuees in outcome $a$. 
% In our analysis, we only consider equilibrium outcomes where all players get a utility greater than $-M_1$.
% \subsubsection{Lower-bound for $\rho(\Gamma_{all})$}
% \label{sec:poa_lower_bound}
Now, we prove the following lemma:
\begin{lemma}
\label{lemma:poa_lower_bound}
There exists instances $\gamma$ of \evacuationgame{}, where
$\rho(\gamma)$ is $O(n)$. Here, $n$ denotes the number of nodes in the evacuation network.
\end{lemma}
\begin{proof}
We present the construction of a specific instance of \evacuationgame{} where the price of anarchy is $O(n)$.
Figure \ref{fig:poa_unsplt_const_lb} shows the construction process. We refer to Figure (\ref{fig:poa_unsplt_const_level_2}) as `level 2', Figure (\ref{fig:poa_unsplt_const_level_3}), (\ref{fig:poa_unsplt_const_level_4}) are referred to 
as `level 3' and `level 4', respectively. In all of these, every edge has a capacity of $1$, and each source has $d$ number of evacuees.

First, let us focus on \textbf{level 2} (Figure \ref{fig:poa_unsplt_const_level_2}). The left figure shows an equilibrium outcome where the players use the red routes. The departure time schedule is as follows: first evacuees under player $0$ evacuate (one evacuee departs at every timestep starting from timestep $0$); once all evacuees under player $0$ have departed, then evacuees under player $1$ evacuate (one evacuee at every timestep starting from timestep $d$). This is an equilibrium because: $(i)$ none of the players can change their route or departure time schedule unilaterally and get a better utility. The cost in this equilibrium would be: $\left(\sum_{t=0}^{d-1}t+3 \right) + \left(\sum_{t=d}^{2d-1}t+3\right) = \sum_{t=0}^{2d-1}t+3 = 2d^2 + 5d$.

Now, the right figure shows a socially optimal outcome, where the players use the blue routes. Here, evacuees under both player $0$ and $1$ start departing at timestep $0$ (one evacuee from each player at each timestep). All evacuees depart by timestep $(d-1)$. The cost in this outcome: $2\sum_{t=0}^{d-1}t+3 = d^2 + 5d$.

So, price of anarchy: $\rho = \frac{2d^2 + 5d}{d^2 + 5d} = \frac{2 + 5/d}{1 + 5/d}$. As $d$ grows larger, we have $\rho \rightarrow 2$.

Now, let us look at \textbf{level 3} (Figure \ref{fig:poa_unsplt_const_level_3}). Through similar analysis, as done for level 2, the cost at equilibrium shown in left figure (red routes) is: $\sum_{t=0}^{3d-1}t+5 = (9d^2 + 27d)/2$. Cost at the socially optimal outcome shown in right figure (blue routes) is: $3\sum_{t=0}^{d-1}t+5 = (3d^2 + 27d)/2$. So, price of anarchy in level 3 instance is: $\rho = \frac{(9d^2 + 27d)/2}{(3d^2 + 27d)/2} = \frac{3 + 9/d}{1 + 9/d}$. As $d$ grows larger, $\rho \rightarrow 3$.

We can extend this to \textbf{level 4}, as shown in Figure \ref{fig:poa_unsplt_const_level_4} and so on. 
For a general \textbf{level $l (\geq 2)$}, cost at equilibrium is: $\sum_{t=0}^{ld-1}t+2l-1 = (l^2d^2 + 4l^2d - 3ld)/2$. Cost at socially optimal outcome is: $l\sum_{t=0}^{d-1}t+2l-1 = (ld^2 + 4l^2d - 3ld)/2$. Price of anarchy is: $\rho = \frac{(l^2d^2 + 4l^2d - 3ld)/2}{(ld^2 + 4l^2d - 3ld)/2} = \frac{l + (4l-3)/d}{1 + (4l-3)/d}$. When $d >> l, \rho \rightarrow l$, i.e., $\rho$ is in $O(l)$. Now, number of nodes in the network in level $l$ is: $n = 4l - 2$. So, $l = (n+2)/4$. So, $\rho$ is in $O(n)$. 
% Because this is just one instance of our game, the price of anarchy of this instance is a lower bound for the price of anarchy over all instances, i.e., $\rho(\Gamma_{all})$. Therefore, $\rho(\Gamma_{all})$ is $\Omega(n)$.
\end{proof}

% The following can be deduced directly from Lemma \ref{lemma:poa_lower_bound}:
% \begin{theorem}
% \label{theorem:poa_lower_bound} 
% Let $\Gamma_{all}$ denote the set of all possible instances of \evacuationgame{}. Then, $\rho(\Gamma_{all})$ is $\Omega(n)$. Here, $n$ denotes the number of nodes in the evacuation network.
% \end{theorem}
% \begin{proof}
% The instance constructed in the proof Lemma \ref{lemma:poa_lower_bound} belongs to the set $\Gamma_{all}$. Therefore, the price of anarchy of this instance is a lower bound for the price of anarchy over all instances, i.e., $\rho(\Gamma_{all})$. Therefore, $\rho(\Gamma_{all})$ is $\Omega(n)$.
% \end{proof}

Lemma \ref{lemma:poa_lower_bound} shows that, it is possible to construct instances of \evacuationgame{} where the price of anarchy is high (i.e., $O(n)$). However, these instances are constructed by carefully choosing the network structure and the edge travel time attributes. 
%It is unlikely to find such road-networks in the real-world. 
In our experiments (Section \ref{sec:experiment_results}), we study two real-world road-networks. 

\subsection{Finding Pure Strategy Nash Equilibrium}
\label{sec:eq_calc}
\begin{figure}[!b]
    \begin{subfigure}[b]{0.235\textwidth}
        \centering
        \includegraphics[width=0.8\textwidth]{images/confluent-routes.pdf}
        \caption{Confluent routes.}
        \label{fig:confluent_routes}
    \end{subfigure}
    \hfill
    \begin{subfigure}[b]{0.235\textwidth}
        \centering
        \includegraphics[width=0.8\textwidth]{images/non-confluent-routes.pdf}
        \caption{Non-confluent routes.}
        \label{fig:non_confluent_routes}
    \end{subfigure}
    \caption{Example of confluent and non-confluent routes.}
    % \vspace{-4mm}
\end{figure}
From Corollary \ref{corollary:potential_game}, we find that \evacuationgame{} has a generalized ordinal potential. This implies that if we follow the best-response dynamics starting from an initial outcome, we will reach an equilibrium eventually \cite{MONDERER1996124}. However, this does not guarantee that we will reach an equilibrium in polynomial time (e.g. it may take exponential amount of time). In this section, we present a \emph{polynomial-time} algorithm for finding equilibria (in \evacuationgame{} instances) under a particular constraint: the routes of the players need to be \textit{confluent}. We call it the Confluent Route Constraint (\confluentconstr{}).
A set of routes is \textit{confluent} when the following is true: if any two routes within the set meet at a node, then the remaining parts of those routes are identical (see example in Figure \ref{fig:confluent_routes}). Confluent routes have been used extensively in the evacuation planning literature due to their practical advantages (detailed discussion in Section \ref{sec:related_works}). Even and Pillac~\emph{et al.}~\cite{even2015convergent} showed that, if there is a path from every source to at least one of the safe nodes, then  there exists a set of routes that are confluent. As a result, Theorem \ref{theorem:eq_existence} is valid under \confluentconstr{}. Moreover, Lemma \ref{lemma:poa_lower_bound} is also valid under \confluentconstr{}; in Figure \ref{fig:poa_unsplt_const_lb} all of the red and blue routes are confluent.

We now present our method for finding equilibria that satisfies \confluentconstr{}, in \evacuationgame{} instances. First, we present an algorithm to calculate the best response of a player. We then use it to find equilibria.

\subsubsection{Calculating Best Response}
% To solve Problem \ref{prob:equilibrium_calc}, we first approach the problem of finding the best response of a player. 
Let $a_{-i}$ denote the action of all players except $i$ in outcome $a$. Then, best response is defined as follows: 
\begin{definition}
Let $a_{-i}$ be fixed. Then, $a_i^*$ is called a best response to $a_{-i}$ if:
$$u_i(a_i^*, a_{-i}) \geq u_i(a_i, a_{-i}), \forall a_i \in A_i$$
\end{definition}

We present the Best Response Strategy Algorithm (\bestresponsealgo{}, Algorithm \ref{alg:solve_brp}) to find the best response of a player in \evacuationgame{}, under the Confluent Route Constraint. 
Here, we provide a proof sketch of the correctness of \bestresponsealgo{} (complete proof of correctness provided in the full version of the paper \cite{islam2024strategic}).
\begin{algorithm}[!b]
    \caption{Best Response Strategy Algorithm (\bestresponsealgo{}). Time complexity: $O(n^2 + m^2\log n + Mmn)$}
    \label{alg:solve_brp}
    \KwIn{ 
        Player $i$,
        Action of all the players except $i$: $a_{-i}$, 
        Evacuation network: $\cG(\cV, \cA)$
    }
    \KwOut{Best response of player $i$: $a_i^*$}
    $\cF \gets$ Directed forest constructed using routes in $a_{-i}$. \label{algline:construct_forest_brp}\\
    $\cG^\prime(\cV, \cA^\prime) \gets$ Copy of $\cG(\cV, \cA)$. \label{algline:net_constr_start}\\
    \smallskip
    \For {each edge $e(u, v) \in \cG^\prime$}{
        \If {$u \in \cF$}{
            Remove edge $e$ from $\cG^\prime$.
        }
    } \label{algline:net_constr_end}
    \smallskip
    $A_{cand} = \{\}$\\
    $x \gets$ Source node of Player $i$.\\
    $C \gets$ Set of distinct edge capacity values from $\cG^\prime$. \label{algline:all_capacity_start}\\
    \smallskip
    \For{each $c \in C$}{ \label{algline:all_capacity_end}
        $\cA^\prime_c = \{e | e \in \cA^\prime\ \text{ and } c_e \geq c\}$\\
        Construct the network $\cG^\prime_c(\cV, \cA^\prime_c)$.\\
        Calculate shortest travel time path from $x$ to all other nodes in $\cG^\prime_c$. \label{algline:shortest_path}\\
        $P_{disjoint} \gets$ Set of shortest travel time paths in $\cG^\prime_c$ from $x$ to any safe node $z \in \cS$. \label{algline:case_1_start}\\
        \For{each path $p \in P_{disjoint}$}{
            $schedule_p \gets$ Schedule on path $p$ with capacity $c$ that maximizes $u_i$. \label{algline:best_schedule_edge_disjoint}\\
            Add action $a = (p, schedule_p)$ to $A_{cand}$ with $u_i(a)$.
        }\label{algline:case_1_end}
        \medskip
        $\cT_\cF \gets$ Set of transit nodes $w$ (i.e. $w \in \cT$) where $w \in \cF$. \label{algline:case_2_start}\\
        $P_{overlap} \gets$ Set of shortest travel time paths in $\cG^\prime_c$ from $x$ to any node $w \in \cT_\cF$.\label{algline:case_2_route_to_transit_calc}\\
        \For{each path $p_{xw} \in P_{overlap}$}{
            $p_{wz} \gets$ Path from transit node $w$ to safe node $z$ in $\cF$.\\
            $p_{xwz} \gets$ Path from $x$ to $z$ through $w$.\\
            $schedule_p \gets$ Schedule on path $p_{xwz}$, with capacity $c$ on $p_{xw}$, that maximizes $u_i$. \label{algline:best_schedule_edge_overlap}\\
            Add action $a = (p, schedule_p)$ to $A_{cand}$ with $u_i(a)$.
        }\label{algline:case_2_end}
    }
    \smallskip
    $a^*_{i} \gets$ Action with maximum utility $u_i(a)$ from $A_{cand}$. \label{algline:highest_utility_start}\\
    \Return $a^*_i$\label{algline:highest_utility_end}
\end{algorithm}
In \bestresponsealgo{}, we have the following assumption:
\begin{assumption}
\label{assumption:existing_routes_confluent}
With $a_{-i}$, all players $j \in [N], j \neq i$ get a utility greater than $-M_1$ and their routes are confluent.
\end{assumption}

\begin{figure}[!b]
    \centering
    \begin{subfigure}[b]{0.33\textwidth}
        \centering
        \includegraphics[width=0.9\textwidth]{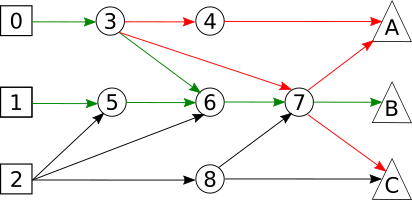}
        \caption{Evacuation Network: $\cG(\cV,\cA)$.}
        \label{fig:brp_1}
    \end{subfigure}
    % \hfill
    \par\bigskip
    \begin{subfigure}[b]{0.33\textwidth}
        \centering
        \includegraphics[width=0.9\textwidth]{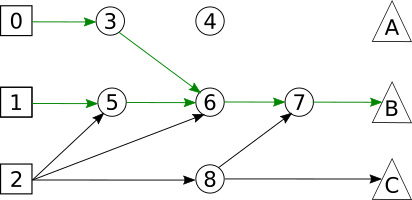}
        \caption{$\hat{\cG}(\cV, \hat{\cA})$}
        \label{fig:brp_2}
    \end{subfigure}
    % \hfill
    \par\bigskip
    \begin{subfigure}[b]{0.33\textwidth}
        \centering
        \includegraphics[width=0.9\textwidth]{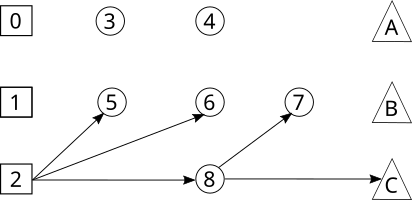}
        \caption{$\cG^\prime(\cV, \cA^\prime)$}
        \label{fig:brp_3}
    \end{subfigure}
    \caption{Proof sketch for \bestresponsealgo{}. The \textit{green} edges denote the routes in $a_{-i}$.}
    \label{fig:brp_example}
    % \vspace{-2mm}
\end{figure}

% Assumption \ref{assumption:existing_routes_confluent} implies that the routes of all the players (except $i$) are confluent with each other, and that there is no capacity violation on any of the edges at any time.
In \bestresponsealgo{}, first we ensure that player $i$'s route is confluent with the routes of the other players. Because of Assumption \ref{assumption:existing_routes_confluent}, the existing routes are already confluent. Therefore, these routes will form either a directed tree rooted at a safe node or a directed forest (i.e., a set of disjoint directed trees). An example is shown in Figure \ref{fig:brp_1}. Here, Player $0$ and $1$ have already chosen their routes 
% ($0 \rightarrow 3 \rightarrow 6 \rightarrow 7 \rightarrow B$, and $1 \rightarrow 5 \rightarrow 6 \rightarrow 7 \rightarrow B$ respectively) 
colored in \textit{green}. Any route of Player $2$ cannot use the \textit{red} edges as that will cause violation of \confluentconstr{}. So, we can discard these \textit{red} edges to get a new network $\hat{\cG}(\cV, \hat{\cA})$ (Figure \ref{fig:brp_2}). 
Now, for routes of Player $i$, there are two possible cases:
\begin{itemize}[leftmargin=*]
\item \textbf{Case 1:} Player $i$'s route is edge disjoint from all the other routes. An example is the route $2 \rightarrow 8 \rightarrow C$ in Figure \ref{fig:brp_2}. To find such routes, we can remove all route edges (i.e., \textit{green} edges) from $\hat{\cG}$ which gives us the network $\cG^\prime$(Figure \ref{fig:brp_3}). We then consider available paths from player $i$'s source node (node $2$ in Figure \ref{fig:brp_3}) to the safe nodes (nodes $A, B, C$ in Figure \ref{fig:brp_3}) in $\cG^\prime$.

\item \textbf{Case 2:} Player $i$'s route can join one of the existing trees at a transit node. Then, the rest of the route (i.e. transit to safe node) is already decided. In our example in Figure \ref{fig:brp_2}, Player $2$'s route can join the existing tree at transit node $5, 6,$ or $7$; the later part of the routes will be $5 \rightarrow 6 \rightarrow 7 \rightarrow B, 6 \rightarrow 7 \rightarrow B$, and $7 \rightarrow B$ respectively. Note that, there can be multiple paths from Player $i$'s source to a particular transit node. We can find all of them from $\cG^\prime$ (Figure \ref{fig:brp_3}).
% Paths from $src_i$ to a transit node must be calculated in $\cG^\prime$ to ensure that these paths do not use any of the \textit{red} edges (preventing violation of confluent flow constraint), or any of the \textit{green} edges.
\end{itemize}

The above two cases cover all the possible routes for Player $i$. In \bestresponsealgo{}, we construct the network $\cG^\prime$ in lines (\ref{algline:net_constr_start}--\ref{algline:net_constr_end}), which is then used in the later part of the algorithm.
Now, the action of Player $i$ also includes a departure time schedule. Once a path is chosen, the best schedule on this path depends on the capacity of the path. 
% For instance, there can be paths with high capacity but also high amount of travel time. On the other hand, there can also be paths with low capacity but small amount of travel time. 
In \bestresponsealgo{}, we consider paths of all possible capacity values (line \ref{algline:all_capacity_start}, and \ref{algline:all_capacity_end}). For each distinct capacity value, we only consider the shortest travel time paths (line \ref{algline:shortest_path}) of that capacity. In the full version of the paper \cite{islam2024strategic}, we show that it suffices to consider the shortest paths only, at each capacity value. There, we also provide details on how to calculate the best schedule on a path.

% The edge disjoint paths (\textit{Case $1$}) are considered in lines (\ref{algline:case_1_start}--\ref{algline:case_1_end}); paths that join the existing routes at some transit node (\textit{Case $2$}) are considered in lines (\ref{algline:case_2_start}--\ref{algline:case_2_end}). 
Paths of Case 1 and 2 are considered in lines (\ref{algline:case_1_start}--\ref{algline:case_1_end}) and (\ref{algline:case_2_start}--\ref{algline:case_2_end}), respectively.
For each path, we can calculate the schedule that maximizes player $i$'s utility (and therefore avoids any capacity violation). This gives us all the candidates for player $i$'s best response that satisfy the confluent route constraint. We return the action with the highest utility for player $i$ (lines \ref{algline:highest_utility_start}--\ref{algline:highest_utility_end}).
Time complexity of \bestresponsealgo{} is $O(n^2 + m^2\log n + Mmn)$, whereas the input to \bestresponsealgo{} requires space $O(m+n^2+M)$ (proof in the full version of the paper \cite{islam2024strategic}). Here $n,m,M$ are the number of nodes, edges and evacuees in the evacuation network, respectively. 
% So, with respect to the input size, 
Therefore, \bestresponsealgo{} is a polynomial-time algorithm.

% Note that, the input to \bestresponsealgo{} is $\cG$ and $a_{-i}$. For a player $j \neq i$, the departure time schedule requires $O(M_j)$ space. This is because of the following possible scenario: only one evacuee is able to leave at each timestep, requiring a total of $M_j$ timesteps. Therefore, just for the departure time schedule in $a_{-i}$, we would need $O(M)$ space. For the routes and evacuation network, the required space would be $O(n^2)$, and $O(m+n)$ respectively; $O(m+n^2)$ combined. So, the input size is $O(m+n^2+M)$. Therefore, with respect to the input size, \bestresponsealgo{} is a polynomial-time algorithm. 

\subsubsection{Calculating Equilibria}
\begin{algorithm}[!b]
    \caption{Sequential Action Algorithm (\seqactionalgo{}).\\Time complexity: $O(n^3 + m^2n\log n + Mmn^2)$}
    \label{alg:SAA}
    % \SetKwRepeat{Do}{do}{while}
    \KwIn{Instance of \evacuationgame{}: $\cI$, A sequence of the players: $Seq$
    }
    \KwOut{An equilibrium outcome: $a^*$}
    $N \gets$ Number of Players in $\cI$.\\
    \For {each $i \in [1, N]$}{
        $P_i \gets$ $i^{th}$ player in sequence $Seq$.\\
        $a_i^* \gets$ Best action (under \confluentconstr{}) of $P_i$ considering only players $P_j, \forall j \in [1,i-1]$  have chosen their actions. \label{algline:choose_best_response}\\
    }
    $a^* \gets (a_1^*, a_2^*, ..., a_N^*)$\\
    \Return $a^*$
\end{algorithm}

We now present the Sequential Action Algorithm (\seqactionalgo{}, Algorithm \ref{alg:SAA}) to find pure strategy Nash equilibria under the Confluent Route Constraint, in a given \evacuationgame{} instance. The input to \seqactionalgo{} is an instance $\cI$ of \evacuationgame{}, and a sequence of the players in $\cI$. In \seqactionalgo{}, we go over the players in order of the input sequence. When processing the $i^{th}$ player in the sequence, we calculate this player's best response under \confluentconstr{} considering that only the preceding $(i-1)$ players have already chosen their action. After processing all the players this way, the chosen action of the players is returned. We now prove the following theorem:

\begin{theorem}
\label{theorem:saa_correctness}
\seqactionalgo{} returns a pure strategy Nash equilibrium under the Confluent Route Constraint, where all players get a utility greater than $-M_1$.
% \seqactionalgo{} returns a pure strategy Nash equilibrium $a^*$ where $u_i(a^*) \geq -M_2, \forall i \in [N]$.
\end{theorem}
\begin{proof}
First we prove that \seqactionalgo{} returns an equilibrium under the confluent route constraint. We prove it by contradiction. 

Let $seq$ denote the sequence of players used by \seqactionalgo{}, and $\hat{a}$ denote the outcome returned by \seqactionalgo{}. Let's assume that $\hat{a}$ is not an equilibrium under \confluentconstr{}. Then, there exists a player $i$ who can unilaterally deviate from $\hat{a}$ and get a higher utility, while satisfying the \confluentconstr{}. Let the better action for player $i$ be denoted by $a^\prime_i$. Let $a^\prime = (a^\prime_i, \hat{a}_{-i})$. Then, 
\begin{align}
    \label{ineq:better_action}
    u_i(a^\prime) > u_i(\hat{a}) \geq -M_1  
\end{align}
% We show that, if such an action $a^\prime_i$ existed, then \seqactionalgo{} would have chosen that action for player $i$ instead of $\hat{a}_i$.
Let $P_{before}$ denote the set of players that come before player $i$ in $seq$. Let's consider the action $a^\prime_i$, compared to $\hat{a}_i$, for player $i$ at her turn in \seqactionalgo{}. 

\begin{enumerate}[leftmargin=*]
\item[I)] As $a^\prime_i$ satisfies \confluentconstr{} with all players, it satisfies \confluentconstr{} with players in $P_{before}$.
\item[II)] From (\ref{ineq:better_action}), we have $u_i(a^\prime) > -M_1$. This means, with all players present $(i)$ $a^\prime_i$ does not cause any capacity violation on player $i$'s route, and $(ii)$ all evacuees under player $i$ reach a safe node within $T_{max}$. Then, $(i)$ and $(ii)$ are also true if $P_{before} \cup \{i\}$ were the only players in the game. Moreover, player $i$ would get the same utility, equal to $u_i(a^\prime)$. 
\end{enumerate}

(I) and (II), together with (\ref{ineq:better_action}), implies that $\hat{a}_i$ is not the best action of player $i$  under \confluentconstr{} during her turn in \seqactionalgo{} (as $a^\prime_i$ is better). That is a contradiction (\seqactionalgo{} line \ref{algline:choose_best_response}).

Now, we also need to prove that with outcome $\hat{a}$, every player gets a utility greater than $-M_1$. Note that, by using \bestresponsealgo{} within $\seqactionalgo{}$, we ensure that there is no capacity violation. So, we only need to prove that, with $\hat{a}$, all evacuees reach a safe node within $T_{max}$. We prove this by utilizing Assumption \ref{assumption:completion_time_upper_bound} and a similar argument provided in the proof of Lemma \ref{lemma:cmpletion_time_upper_bound}. The complete proof is provided in the full version of the paper \cite{islam2024strategic}.
% For the first player $P_1$ in $seq$, no matter what route is chosen in $\hat{a}$, a possible departure time schedule for $P_1$ is as follows: starting from timestep $0$, one evacuee will leave at every timestep. This way all evacuees under $P_1$ will depart by timestep $d_1-1$. In $\hat{a}$, all evacuees under $P_1$ will also depart by timestep $d_1-1$ (otherwise $\hat{a}_1$ would not be $P_1$'s best action within \seqactionalgo{}). So, in $\hat{a}$, all evacuees under $P_1$ will reach a safe node within $(\tau + d_1 - 1)$, where $\tau = \sum_{e \in \cA}\tau_e$. Now for player $2$, no matter what route is chosen in $\hat{a}$, following is a possible departure time schedule: starting from timestep $\tau + d_1$, one evacuee under $P_2$ leaves at every timestep, the last evacuee leaving at timestep $\tau + d_1 + d_2 -1$. By similar argument, in $\hat{a}$, all evacuee will also depart by timestep $\tau + d_1+d_2-1$, and therefore, all evacuees under $P_1, P_2$ will reach safe node within timestep $2\tau + d_1+d_2-1$. Continuing this up to the last player $P_N$, in $\hat{a}$, all evacuees will reach a safe node within timestep $N\tau + d_1 + d_2 + ... + d_N -1= N\tau + M - 1 <= n\tau + M - 1 = T_{max}$.
\end{proof}
% \seqactionalgo{} provides us equilibria in a given \evacuationgame{} instance under a particular constraint, i.e. the Confluent Route Constraint (\confluentconstr{}). \confluentconstr{} has practical value as it gives us confluent routes that are beneficial in real-world evacuations. In addition, it allows us to compute equilibria quickly, i.e., in polynomial-time.

\section{Empirical Evaluation}
\label{sec:experiment_results}
\begin{figure}[t]
    \centering
    \begin{subfigure}[b]{0.49\textwidth}
        \includegraphics[width=\textwidth]{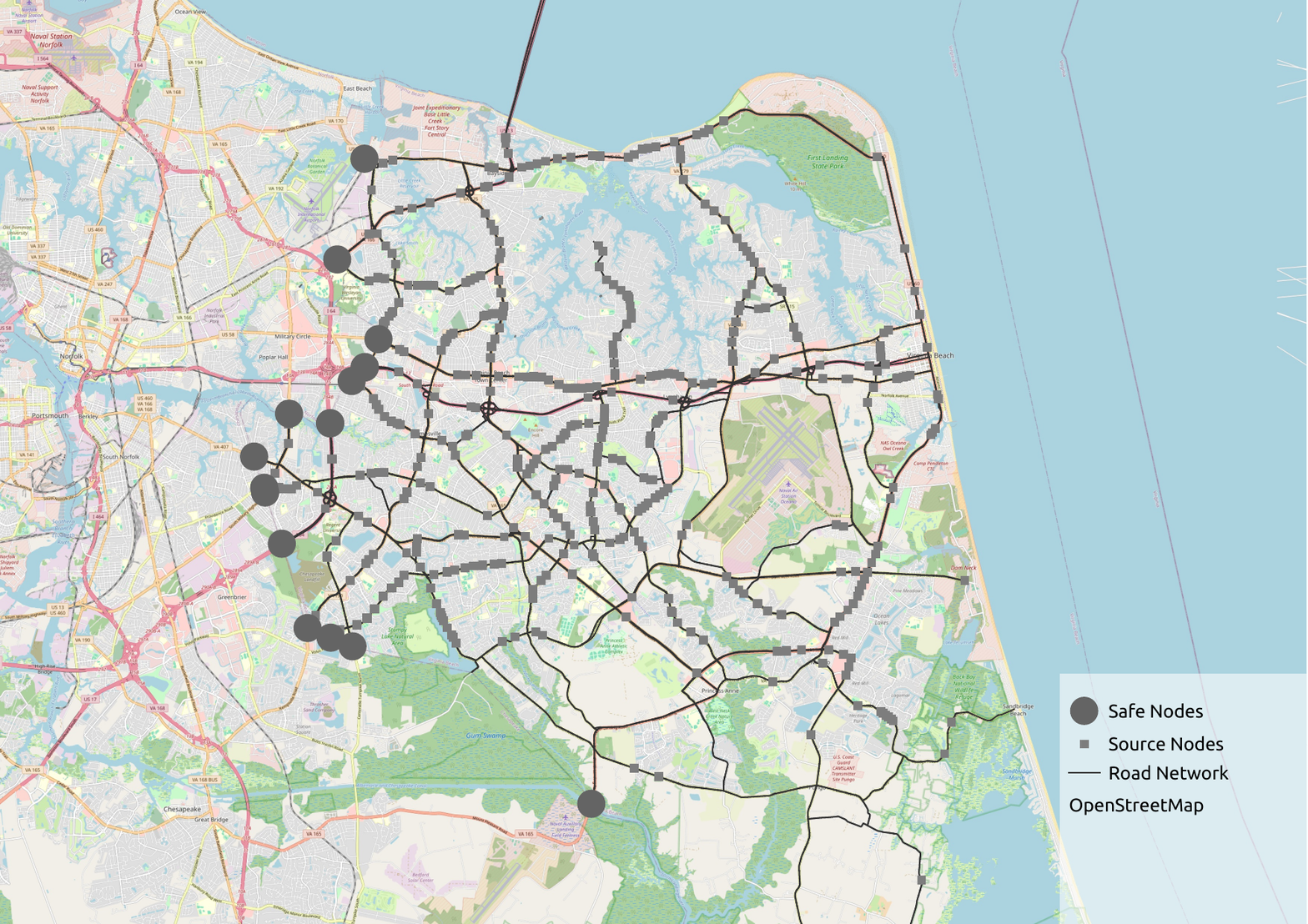}
        \caption{Evacuation network $(\cG)$.}
        \label{fig:vbc}
    \end{subfigure}
    \par\bigskip
    \begin{subfigure}[b]{0.5\textwidth}
        \centering
        % \small
        \begin{tabular}{p{5cm}p{2.25cm}}
        \toprule
        \# of nodes $(n)$ and edges $(m)$ in $\cG$ & 3522, 7003\\
        \midrule
        \# of players & 501\\
        \midrule
        \# evacuees (i.e. households) $(M)$ & $165$ Thousand\\
        \midrule
        Evacuation time upper bound $(T_{max})$ & $8$ hours\\
        \midrule
        Length of one timestep & $0.5/1/2$ minutes\\
        \bottomrule
        \end{tabular}
        \caption{Game instance details.}
        \label{tab:hc}
    \end{subfigure}
    \caption{Virginia Beach City game instance.}
    \label{fig:vbc_game}
    % \vspace{-2mm}
\end{figure}

\begin{figure}[t]
    \centering
    \begin{subfigure}[b]{0.49\textwidth}
        \includegraphics[width=\textwidth]{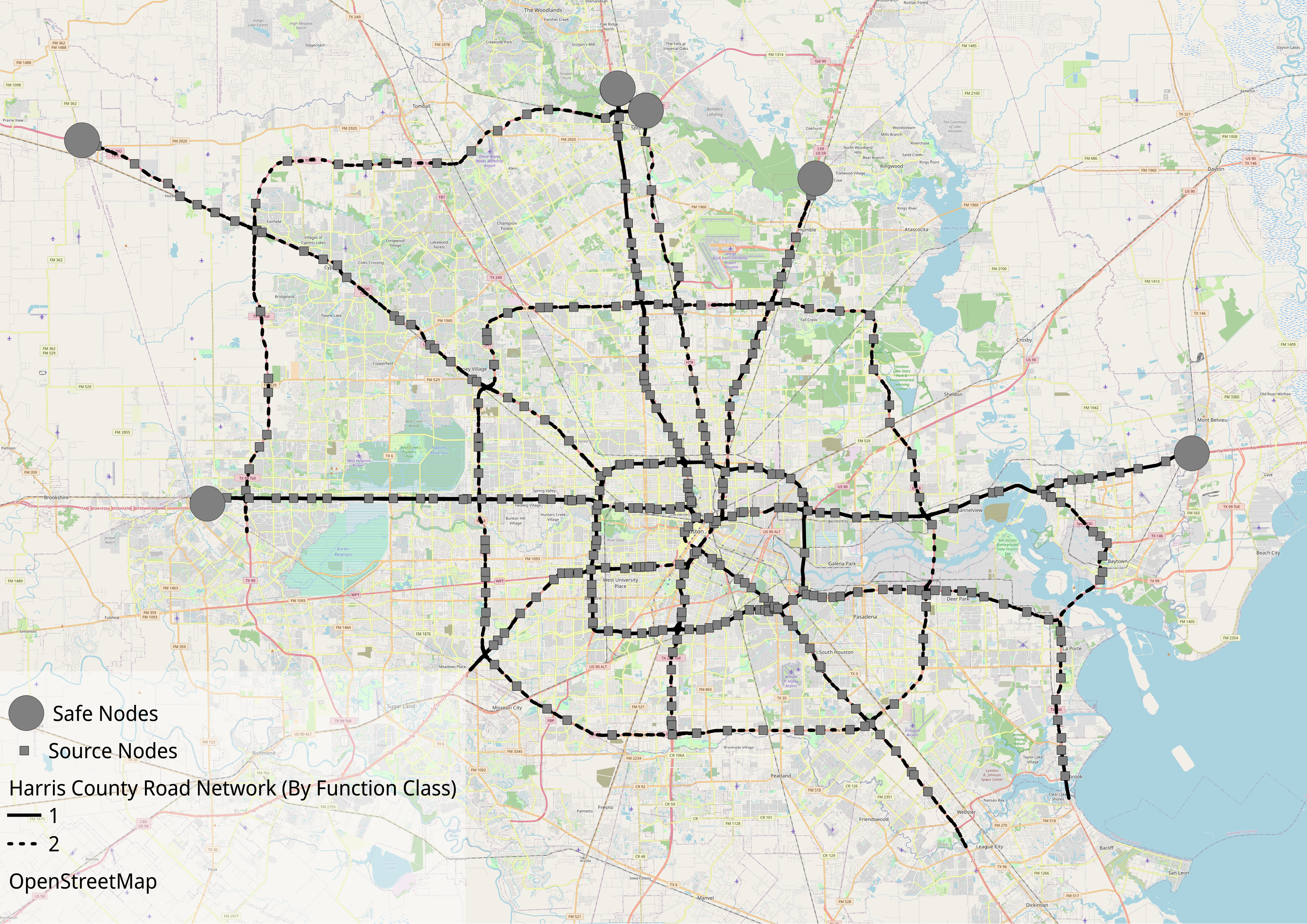}
        \caption{Evacuation network$(\cG)$.}
        \label{fig:hc}
    \end{subfigure}
    % \smallskip
    \par\bigskip
    \begin{subfigure}[b]{0.5\textwidth}
        \centering
        % \small
		% \begin{tabular}{p{2.625cm}p{0.75cm}p{2.7cm}p{0.7cm}}
  %       \toprule
  %       \textbf{Property} & \textbf{Value} & \textbf{Property} & \textbf{Value}\\
  %       \toprule
  %       \# of nodes $(n)$ in $\cG$ & $1338$ & \# of edges $(m)$ in $\cG$ & $1751$\\
  %       \midrule
  %       \# of players $(N)$ & $374$ & \# of evacuees $(M)$ & $1.5$M$^*$\\
  %       \bottomrule
  %       \end{tabular}
  %       \begin{tabular}{p{5.55cm}p{2cm}}
  %       \toprule
  %       \textbf{Property} & \textbf{Value}\\
  %       \toprule
  %       Evacuation time upper bound $(T_{max})$ & $15$ hours\\
  %       \midrule
  %       Length of one timestep & $2$ minutes\\
  %       \bottomrule
  %       \end{tabular}
        \begin{tabular}{p{5cm}p{2.25cm}}
        \toprule
        \# of nodes $(n)$ and edges $(m)$ in $\cG$ & 1338, 1751\\
        \midrule
        \# of players & 374\\
        \midrule
        \# evacuees (i.e. households) $(M)$ & $1.5$ Million\\
        \midrule
        Evacuation time upper bound $(T_{max})$ & $15$ hours\\
        \midrule
        Length of one timestep & $0.5/1/2$ minutes\\
        \bottomrule
        \end{tabular}
        \caption{Game instance details.}
        \label{tab:hc}
    \end{subfigure}
    \caption{Harris County game instance.}
    \label{fig:hc_game}
\end{figure}
To evaluate the effectiveness of \seqactionalgo{} in real-world scenarios, we consider two study areas: $(i)$ The City of Virginia Beach in Virginia, and $(ii)$ Harris County in Houston, Texas. 
To construct evacuation network for these two areas, we use road network data from HERE maps~\cite{HERE:20}, and synthetic population data presented by Adiga~\emph{et al.}~\cite{adiga15US}. To define the source nodes, we use the household location of the residents of these areas and map them to the nearest node in the road network. We assume that each household evacuates in a single vehicle, and therefore, we consider them as a single evacuee. For safe nodes, we consider locations on major highways at the periphery of the two areas. 
A visualization and summary of the constructed \evacuationgame{} instances are provided in Figure \ref{fig:vbc_game}, \ref{fig:hc_game}. 
% Further details are provided in the supplementary materials.

We conduct experiments for testing two aspects of our proposed methodology: $(i)$ quality of the equilibria found by \seqactionalgo{}, and $(ii)$ scalability of \seqactionalgo{}. We performed the experiment on a high-performance computing cluster, with 64GB RAM and 4 CPU cores allocated to our tasks.

\subsection{Quality of the Equilibria}
\begin{figure}[t]
    \centering
    \includegraphics[width=0.475\textwidth]{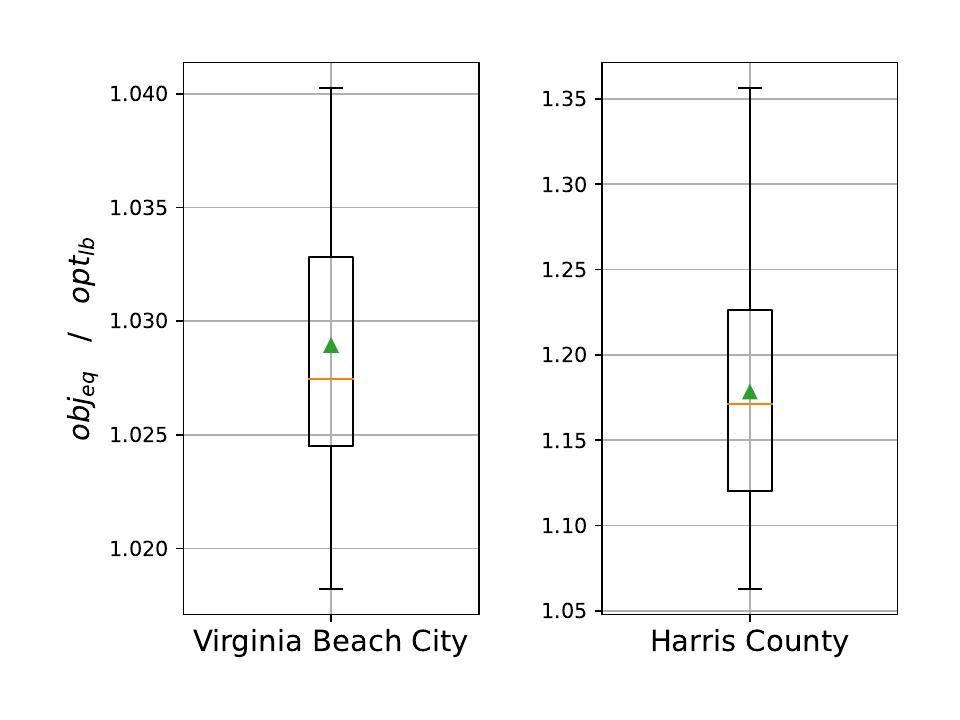}
    \caption{Box-plot showing the ratio of cost at equilibrium to a lower bound of optimal cost, of the $30$ equilibrium outcomes calculated by \seqactionalgo{}, for our two problem instances. The closer the ratio value to $1$, the better.}
    \label{fig:obj_ratio}
    % \vspace{-2mm}
\end{figure}
To calculate equilibria in each of our two problem instances, we first sampled $30$ permutations of the players, uniformly at random. 
We used a timestep of $2$ minutes for constructing our game instances. 
We then provided the permutations and the game instances as input to \seqactionalgo{}. This provided us with $30$ equilibria in each of our two problem instances.
% All of the runs terminated within an hour, and we got 30 different equilibrium outcomes for our game instance.

Since, \seqactionalgo{} provides us equilibria that have confluent routes, we use the optimal confluent routes and schedule as our baseline. We use the Mixed Integer Program (MIP) formulation provided by Islam and Chen~\emph{et al.}~\cite{islam2023optimal} to find the optimal confluent routes and schedule for our two problem instances. We solved the MIPs with the MIP solver Gurobi~\cite{gurobi}.  However, since finding an optimal outcome is {\sf NP-hard} and also hard to approximate~\cite{islam2023optimal}, the solver was not able to provide us with exact optimal solutions. The solver did provide us with lower-bounds on the optimal cost (i.e. lower-bound on $\min_{a \in A} C(a)$). We denote this lower-bound by $opt_{lb}$, i.e. $opt_{lb} \leq \min_{a \in A} C(a)$. Let the cost (i.e. sum of the evacuation time of all evacuees) in an equilibrium outcome be denoted by $obj_{eq}$. We then look at the ratio $obj_{eq} / opt_{lb}$ to quantify the quality of the equilibria found by \seqactionalgo{}. For example, if the ratio is $2$ for an equilibrium outcome, then it means the cost at that equilibrium is at-most twice the cost at an optimal solution. The closer the ratio value to $1$, the better.

Figure \ref{fig:obj_ratio} shows the ratio values using a boxplot. We observe that for Virginia Beach City, the mean value of the ratio is $\sim1.03$ with a standard deviation of $\sim0.006$. For Harris county, the mean ratio value is $\sim1.17$ with a standard deviation of  $\sim0.069$. This indicates that for our two problem instances, \seqactionalgo{} has found equilibria that have social objective close to the optimal social objective.

\subsection{Scalability of \seqactionalgo{}}
\begin{figure}[t]
    \centering
    \begin{subfigure}[b]{0.5\textwidth}
        \begin{tabular}{p{2.5cm}p{1.5cm}p{1.5cm}p{1.5cm}}
        \toprule
        Timestep\newline(minutes) & $2$& $1$& $0.5$\\
        \midrule
        TEG nodes,\newline edges & $801$K, $1.5$M & $1.8$M, $3.3$M & $3.7$M, $6.8$M\\
        \midrule
        MIP solve time\newline (minutes) & $28$ & $ > 120$ & $ > 120$\\
        \midrule
        SAA avg. run\newline time (minutes) & $14$ & $20$ & $27$ \\
        \bottomrule
        \end{tabular}
        \caption{Virginia Beach City.}
        \label{fig:vbc_runtime}
    \end{subfigure}
    \par\bigskip
    \begin{subfigure}[b]{0.45\textwidth}
        \begin{tabular}{p{2.35cm}p{1.6cm}p{1.5cm}p{1.5cm}}
        \toprule
        Timestep\newline(minutes) & $2$ & $1$ & $0.5$\\
        \midrule
        TEG nodes,\newline edges & $685$K, $842$K & $1.4$M, $1.7$M & $2.9$M, $3.5$M\\
        \midrule
        MIP solve time\newline (minutes) & $56$ & $> 120$ & $> 120$\\
        \midrule
        SAA avg. run\newline time (minutes) & $34$ & $36$ & $37$ \\
        \bottomrule
        \end{tabular}
        \caption{Harris County.}
        \label{fig:hc_runtime}
    \end{subfigure}
    \caption{MIP solve time and \seqactionalgo{} run time at different timestep lengths (`K' and `M' denote `Thousand' and `Million', respectively). We observe that, when timestep length is halved, the size of the Time Expanded Graph is roughly doubled. This drastically increases the solve time of the MIP (no feasible solution found in $2$ hours for timestep length $1$ and $0.5$ minute). Impact on the run time of \seqactionalgo{}, on the other hand, is limited.}
    \label{fig:saa_scalability}
    % \vspace{-4mm}
\end{figure}
Existing methods \cite{even2015convergent,romanski2016benders,hafiz2021large,Islam2022fairness,islam2023optimal} for finding evacuation routes and schedule mainly utilize Mixed Integer Program (MIP) formulations of the problem and rely on a Time-Expanded Graph (TEG). This time-expanded graph contains a copy of each node and edge of the evacuation network at each timestep; i.e. the size of the time-expanded graph increases linearly with the number of timesteps. The size of the time-expanded graph determines the number of variables and constraints in the MIPs, and therefore, affects the time it takes to solve the MIP. As a result, existing methods suffer from a scalability issue when the size of the time-expanded graph becomes too large.

To demonstrate the above-mentioned scalability issue of existing methods, we perform experiments with three different timestep lengths: $2$ minute, $1$ minute, and $0.5$ minute. As shown in Figure \ref{fig:saa_scalability}, we see that when the timestep length is halved, the size of the time-expanded graph (in terms of number of nodes and edges) is (roughly) doubled. This drastically increases the solve time of the MIP. In fact, Gurobi was not able to find a feasible solution in $2$ hours for timestep length $1$ and $0.5$ minute. 

\seqactionalgo{}, on the other hand, does not rely on a time-expanded graph. As a result, the impact of using smaller timesteps on the run time of \seqactionalgo{} is limited. This demonstrates that \seqactionalgo{} is more scalable, compared to existing MIP based solution methods.
\section{Conclusion}
In this paper, we have presented a strategic routing and scheduling game (\textit{Evacuation Game: Routing and Scheduling}, \evacuationgame{}) where players choose their routes and departure schedule. We proved that every instance of our game has at-least one pure strategy Nash equilibrium. We also showed that an optimal outcome in an \evacuationgame{} instance will always be an equilibrium in that instance. We provided bounds on the Price of Stability and Price of Anarchy of our game instances. We then presented a polynomial-time algorithm (\textit{Sequential Action Algorithm}, \seqactionalgo{}) for finding equilibria under a practically useful and well-known constraint in the evacuation planning literature, namely the Confluent Route Constraint (\confluentconstr{}). We performed experiments on two study areas: Virginia Beach City in Virginia, and Harris County in Houston Texas. Our experiment results show that \seqactionalgo{} efficiently finds equilibria that have social objective close to the optimal social objective.

\section*{Acknowledgment}
This work was partially supported by University of Virginia Strategic Investment Fund SIF160, the NSF Grants: CCF-1918656, OAC-1916805, RISE-2053013, and the NASA Grant 80NSSC22K1048.

\appendix
\section{Proof of Lemma \ref{lemma:cmpletion_time_upper_bound}}
\label{sec:app_lemma_1_proof}
\textbf{Lemma \ref{lemma:cmpletion_time_upper_bound}:} \textit{Given an evacuation network $\cG(\cV, \cA)$, there exists an outcome where all evacuees arrive at some safe node within time $(n\tau + M - 1)$ without any edge capacity violation. Here, $n$ is the number of nodes in $\cG$, $M$ is total number of evacuees, and $\tau = \sum_{e \in \cA} \tau_e$.}

\begin{proof}
\label{proof:cmpletion_time_upper_bound}
In our formulation, we assume that there is at least one path from every source to at least one of the safe nodes in $\cG$. Therefore, there exists a set of routes that connects each source node to at-least one of the safe nodes.
% In that case, Even and Pillac~\emph{et al.}~\cite{even2015convergent} showed that there exists a set of evacuation routes that are confluent. 
Now, we show that for any given set of evacuation routes, there is a departure time schedule where everyone reaches a safe node within time $(n\tau + M - 1)$ without any edge capacity violation.

Given any set of routes, we can design a departure time schedule as follows: we choose a source $k_1$ and starting from timestep $0$, we let one evacuee depart from it in each timestep. As capacity $c_e \geq 1, \forall e \in \cA$, and there is no other evacuee from the other sources in the network, no capacity violation will occur. Let $d_{k_1}$ denote the number of evacuees at $k_1$. Then the last evacuee will depart at timestep $(d_{k_1} - 1)$. Now, any route, which is a simple path, will have a travel time of at most $\tau = \sum_{e \in \cA} \tau_e$. Therefore, all $d_{k_1}$ evacuees from source $k_1$ will reach safety within time $(\tau + d_{k_1} - 1)$. 

We then choose another source $k_2$ and starting from timestep $(\tau + d_{k_1})$, we let one evacuee depart from $k_2$ at each timestep. Now, because all evacuees from source $k_1$ will already have reached a safe node, and there are no other evacuees in the network, therefore no capacity violation will occur. The last evacuee from $k_2$ will depart at timestep $\tau + d_{k_1} + d_{k_2}-1$ and reach safety within timestep $2\tau + d_{k_1} + d_{k_2} - 1$. We evacuate all the other sources in the same way. Now, the number of source nodes is $N$. So, all $M$ evacuees will reach safety by timestep $N\tau + \sum_{k \in \cE} d_k - 1 = N\tau + M -1 \leq n\tau + M -1$ (as $N \leq n$).

So, there exists a set of evacuation routes and a departure time schedule such that all evacuees arrive at some safe node within timestep $(n\tau + M - 1)$ without any edge capacity violation. We can assign these routes and departure time schedule to the players, which will provide us our desired outcome.
\end{proof}

Note that, Lemma \ref{lemma:cmpletion_time_upper_bound} remains valid even when we require the set of evacuation routes to be confluent. This is because, Even and Pillac~\emph{et al.}~\cite{even2015convergent} proved that there exists a set of evacuation routes that are confluent, when each source node has a path to at-least one safe node. We can then take such a set of confluent routes and construct a departure time schedule as described in the proof of Lemma \ref{lemma:cmpletion_time_upper_bound}. Then all the evacuees will be able to reach a safe node within time $(n\tau + M -1)$ without any edge capacity violation.

\section{\bestresponsealgo{} Details and Proofs}

\subsection{Data Structure for representing a departure time schedule}
\label{sec:schedule_data_structure}
To represent the departure time schedule of a player $i$, we use a list of 2-tuples: $(t, f_t)$. Here, $t$ represents \textit{timestep} and $f_t$ represents number of evacuees under player $i$ who depart at timestep $t$. If no evacuee leaves at a timestep $t^\prime$, then we do not include $(t^\prime, 0)$ in the list. An example of our data structure for departure schedule is provided in Example \ref{example:departure_schedule}.
\begin{example}
\label{example:departure_schedule}
Let's consider the departure time schedule: $\{(0, 10),\\ (1, 10), (3, 5)\}$. Here, 10 evacuees depart at timestep 0 and 1; no evacuee departs at timestep 2; 5 evacuees depart at timestep 3. No evacuee departs after timestep 3.
\end{example}

\subsection{Size of \bestresponsealgo{} Input}
The input to \bestresponsealgo{} are the evacuation network $\cG$, and the action of all players except $i$, i.e. $a_{-i}$. Size of $\cG$ is $O(m+n)$ where, $m, n$ denote the number of nodes, edges in $\cG$, respectively. 

For each player $j \in [N], j\neq i$, the action $a_j$ consists of a path $p_j$ and a schedule $schedule_j$. As $p_j$ is a simple path, it's size is $O(n)$. In $schedule_j$, there can be at most $M_j$ elements (see Section \ref{sec:schedule_data_structure}), where $M_j$ is the number of evacuees under player $j$. So, size of the action $a_j$ is $O(n + M_j)$. The number of players is $N \leq n$. So, the size of $a_{-i}$ is $O((n-1)n + \sum_{j \in [N], j \neq i}M_j)$ or $O(n^2 + M)$ where $M = \sum_{j \in [N]}M_j$ i.e. the total number of evacuees.

So, the total input size is: $O(m+n+n^2+M)$ or $O(m + n^2 + M)$.

\subsection{Calculating the best schedule on a given path for a given capacity}
In \bestresponsealgo{} (Algorithm \ref{alg:solve_brp}), for each distinct capacity $c$ (line \ref{algline:all_capacity_end}), we calculate the best schedule on a given path in two cases: \textit{Case 1:} when player $i$'s path is edge-disjoint from all the other routes (line \ref{algline:best_schedule_edge_disjoint}), and \textit{Case 2:} when player $i$'s path joins an existing tree at a transit node (line \ref{algline:best_schedule_edge_overlap}). We now present how to calculate the best schedule in these two cases.

\begin{algorithm}[t]
    \caption{Calculating best schedule on edge-disjoint path (Case 1). Running Time: $O(M)$}
    \label{alg:best_schedule_edge_disjoint_path}
    \KwIn{ 
        Player $i$,
        Capacity $c$,
        Edge-disjoint path $p$, 
        Evacuation network $\cG$
    }
    \KwOut{Best schedule of player $i$ on path $p$}
    % $c_p \gets \min_{e \in p} c_e$\\
    $d_i \gets$ Number of evacuees under player $i$.\\
    $departure\_schedule = \{\}$\\
    $t \gets 0$\\
    \While{$d_i > 0$}{ \label{algline:disjoint_best_schedule_loop}
        $f_t \gets \min(c, d_i)$\\
        Add $(t, f_t)$ to $departure\_schedule$\\
        $d_i \gets d_i - f_t$\\
        $t \gets t+1$
    }
    \Return $departure\_schedule$
\end{algorithm}

\textbf{Case 1:} Let the edge disjoint path of player $i$ be denoted by $p$. Note that, capacity on path $p$ will be at least $c$, as $p$ is calculated from $\cG_c^\prime$ (line \ref{algline:case_1_start}). As none of the other players use any edge $e \in p$, therefore, the available capacity on all of these edges will be the same at all timesteps. So we can use Algorithm \ref{alg:best_schedule_edge_disjoint_path} to find the best schedule on path $p$. Here, we send $c$ (or the remaining) number of evacuees at each timestep until we have sent out all the evacuees.

\begin{algorithm}[t]
    \caption{Calculating best schedule on a path that joins an existing tree (Case 2).\newline Running Time: $O(M)$}
    \label{alg:best_schedule_edge_overlap_path}
    \KwIn{ 
        Player $i$ (source node denoted by $x$),
        Capacity $c$,
        Transit node $w$,
        Safe node $z$, 
        Evacuation network $\cG$
    }
    \KwOut{Best schedule of player $i$ on path $p_{xwz}$}
    $d_i \gets$ Number of evacuees under player $i$.\\
    $\tau_{xw} \gets$ Travel time on path $p_{xw}$.\\
    $departure\_schedule \gets \{\}$\\
    $t \gets 0$\\
    \While{$d_i > 0$}{ \label{algline:overlap_best_schedule_loop}
        $c^\prime \gets remaining\_capacity[p_{wz}, t+\tau_{xw}]$\\
        $f_t \gets min(c, c^\prime, d_i)$ \label{algline:calc_flow}\\
        \If{$f_t > 0$}{
            Add $(t, f_t)$ to $departure\_schedule$.\\
            $d_i \gets d_i - f_t$
        }
        $t \gets t+1$
    }
    \Return $departure\_schedule$
\end{algorithm}

\textbf{Case 2:} Let's assume that player $i$'s route joins an existing tree at transit node $w$ and then by following the tree the route ends at safe node $z$. Let's denote player $i$'s source node by $x$. Then the path of player $i$ would be $p_{xwz}$. Now, capacity on the path $p_{xw}$ (i.e. path from $x$ to $w$) will be the same at all timesteps (and at least $c$), as the other players do not use any edge on $p_{xw}$ ($p_{xw}$ is calculated from $\cG_c^\prime$ in line \ref{algline:case_2_route_to_transit_calc}). However, the capacity on path $p_{wz}$ can be different at different timesteps, as other players are using edges on $p_{wz}$. We, therefore, need to keep track of the remaining capacity (at different timesteps) on the edges that are being used by the other players. Let, $remianing\_capacity[p, t]$ denote the remaining capacity on path $p$ when entering at timestep $t$. Then, we can calculate the best schedule on $p_{xwz}$ using Algorithm \ref{alg:best_schedule_edge_overlap_path}. Here, starting at timestep 0, we calculate the largest number of evacuees we can send at each timestep (line \ref{algline:calc_flow}), and then send that number of evacuees.

To calculate the remaining capacity on paths (i.e.,\\ $remaining\_capacity[p, t]$), we do some preprocessing with the actions of the other players $a_{-i}$. We first construct the forest $\cF$ of paths from $a_{-i}$. Note that, the path that starts from an edge $e$ in $\cF$ and ends at a safe node, is unique. This is because the paths in $a_{-i}$ are confluent (Assumption \ref{assumption:existing_routes_confluent}). Therefore, if $e_{first}$ is the first edge on $p$, then we can use $remaining\_capacity[e_{first}, t]$ to denote $remaining\_capacity[p, t]$. We use a hash table to save these remaining capacities. Algorithm \ref{alg:preprocess} shows the preprocessing steps.

\begin{algorithm}[t]
    \caption{\textit{preprocess}: preprocessing $a_{-i}$ for \bestresponsealgo{}. Time Complexity $O(n^2 + Mn)$}
    \label{alg:preprocess}
    \KwIn{Action of all the players except $i$: $a_{-i}$, Evacuation network $\cG$.}
    \KwOut{
        Route forest from $a_{-i}$: $\cF$, 
        Remaining capacity on edges: $remaining\_capacity$
    }
    $\cF \gets$ Directed forest constructed using the routes from $a_{-i}$\\
    $remaining\_capacity \gets$ Empty hash table.\\
    \For{each player $j \in [N], j \neq i$}{
        $path_j \gets$ Sequence of edges on path of player $j$ in $a_{-i}$.\\
        $schedule_j \gets$ Departure time schedule of player $j$ in $a_{-i}$.\\
        \For{$(t, f_t)$ in $schedule_j$}{ 
            $path_j(t) \gets$ Empty Sequence.\\
            $t_{entry} \gets t$\\
            \For{$e \in path_j$}{ \label{algline:path_time_loop_start}
                Append $(e, t_{entry})$ to $path_j(t)$.\\
                $t_{entry} \gets t_{entry} + \tau_e$ \label{algline:path_time_loop_end}
            }
            \For{$(e, t_{entry}) \in path_j(t)$}{ \label{algline:update_loop_start}
                \If{$(e, t_{entry}) \notin remaining\_capacity$}{
                    $remaining\_capacity[e, t_{entry}] \gets c_e$
                }
                $remaining\_capacity[e, t_{entry}] \gets remaining\_capacity[e, t_{entry}] - f_t$ \label{algline:update_capacity} \label{algline:update_loop_end}\\
            }
            $(e, t_{entry}) \gets$ last element in sequence $path_j(t)$ \label{algline:backward_update_start}\\
            backward\_update(\newline$\cF, e, t_{entry},$\newline$remaining\_capacity[e, t_{entry}],$\newline$remaining\_capacity$\newline) \label{algline:backward_update_end}
        }
    }
    \Return $\cF, remaining\_capacity$
\end{algorithm}

\begin{algorithm}[t]
    \caption{\textit{backward\_update}: backward update of path capacity in preprocess. Time Complexity $O(n)$}
    \label{alg:backward_update}
    \KwIn{
        Forest: $\cF$, 
        Edge: $e(u, v)$,
        Edge entry time: $t_{entry}$,
        Path capacity: $c_p$, 
        Remaining capacity on edges: $remaining\_capacity$
    }
    \KwOut{Updated remaining capacity on edges: $remaining\_capacity$}
    \For{$e^\prime(w, u) \in \delta^-_{\cF}(u)$}{
        $t^\prime_{entry} \gets t_{entry} - \tau_{e^\prime}$\\
        % $e^\prime \gets (w_{t^\prime}, u_{t_1})$\\
        \If{$(e^\prime, t^\prime_{entry}) \notin remaining\_capacity$}{
            $remaining\_capacity[e^\prime, t^\prime_{entry}] \gets c_{e^\prime}$
        }
        $remaining\_capacity[e^\prime, t^\prime_{entry}] \gets \min(c_p, remaining\_capacity[e^\prime, t^\prime_{entry}])$ \label{algline:backward_update_capacity}\\
        backward\_update(\newline$\cF, e^\prime, t^\prime_{entry},$\newline$remaining\_capacity[e^\prime, t^\prime_{entry}],$\newline$remaining\_capacity$\newline)
    }
    
\end{algorithm}

In Algorithm \ref{alg:preprocess}, we go over each player $j$, excluding player $i$. Then, based on player $j$'s schedule, we update the remaining capacity of edges on player $j$'s route at appropriate timesteps (line \ref{algline:update_capacity}). When we have updated the capacity of edge $e$ in the forest $\cF$, we have to update the remaining capacity of all edges $e^\prime$ that leads to $e$ in $\cF$; this is because the capacity on edge $e$ limits the capacity of all paths $p$ that go through $e$. We, therefore, do a backward traversal of the forest $\cF$ starting from the last edge on player $j$'s route (Algorithm \ref{alg:preprocess} line \ref{algline:backward_update_start}--\ref{algline:backward_update_end}), and update the capacity values appropriately (Algorithm \ref{alg:backward_update}) line \ref{algline:backward_update_capacity}).

In both Algorithm \ref{alg:best_schedule_edge_disjoint_path} (Case 1) and \ref{alg:best_schedule_edge_overlap_path} (Case 2), for a given path $p$ and capacity $c$, we send the largest number of evacuees possible at every timestep. Therefore, they provide us the best schedule for the given path and capacity. 

\subsection{Proof of Correctness of \bestresponsealgo{}}
We prove the following lemma: 
\begin{lemma}
\label{lemma:brsa_correctness}
Let $a_i$ denote the action returned by \bestresponsealgo{} and $b_i$ denote a best response action of player $i$. Then: $u_i(a_i, a_{-i}) = u_i(b_i, a_{-i})$. 
\end{lemma}
\begin{proof}
As $b_i$ is a best response of player $i$, therefore, by definition: $u_i(b_i, a_{-i}) \geq u_i(a_i, a_{-i})$. So, we only need to prove that: $u_i(a_i, a_{-i}) \geq u_i(b_i, a_{-i})$.

Let $x$ denote the source node of player $i$. There are two possible cases for the path in $b_i$; \textit{Case 1:} the path in $b_i$ is disjoint from all the other routes, and \textit{Case 2:} the path in $b_i$ joins an existing tree of paths at a transit node. 
% Therefore, $b_i$ will either use a path of \textit{Case 1} or a path of $Case 2$. 
Let us consider these two possibilities individually.

\textbf{Case 1:} Let, $p_b$ denote the path in $b_i$. Let the travel time and capacity on path $p_b$ be denoted by $\tau_b, c_b$ respectively. Let's also assume that $p_b$ ends at safe node $z$. Now, \bestresponsealgo{} will consider the capacity $c_b$ (Algorithm \ref{alg:solve_brp} line \ref{algline:all_capacity_end}) and consider a shortest travel time path $p_{a^\prime}$ from $x$ to $z$ that has capacity at least $c_b$. Let, the travel time on path $p_{a^\prime}$ be denoted by $\tau_{a^\prime}$. Then, $\tau_{a^\prime} \leq \tau_b$. \bestresponsealgo{} will calculate the best schedule $schedule_{a^\prime}$ for path $p_{a^\prime}$ (using Algorithm \ref{alg:best_schedule_edge_disjoint_path}) and capacity $c_b$ and then add the action $a^\prime_i = (p_{a^\prime}, schedule_{a^\prime})$ to $A_{cand}$ (Algorithm \ref{alg:solve_brp} line \ref{algline:case_1_end}). Now, because in both $a^\prime_i$ and $b_i$, the same path capacity $c_b$ is used, and because $\tau_{a^\prime} \leq \tau_{b}$, so: $$u_i(a^\prime_i, a_{-i}) \geq u_i(b_i, a_{-i})$$

\textbf{Case 2:} Here, we assume that the path in $b_i$ joins an existing tree of paths at a transit node. Let this transit node be denoted by $w$. Let $p_b$ denote the path in $b_i$, that starts at source node $x$, goes through transit node $w$, and ends at safe node $z$. Now, $p_b$ has two parts: path from $x$ to $w$ (denoted by $p^{xw}$) and path from $w$ to $z$ (denoted by $p^{wz}$). The edges on $p^{xw}$ are not used by other players, therefore, the capacity on path $p^{xw}$ will be same at all timesteps. Let, the capacity of path $p^{xw}$ be denoted by $c^{xw}_b$. Also, let $\tau^{xw}_b$ denote the travel time on path $p^{xw}$.
Now, \bestresponsealgo{} will consider a shortest path from $x$ to $w$ (let use denote it by $p^{xw}_{a^\prime}$) that has capacity at least $c^{xw}_b$ (Algorithm \ref{alg:solve_brp} line \ref{algline:all_capacity_end}). Let the travel time on path $p^{xw}_{a^\prime}$ be $\tau^{xw}_{a^\prime}$. Then, $\tau^{xw}_{a^\prime} \leq \tau^{xw}_b$. 
\bestresponsealgo{} will calculate the best schedule $schedule_{a^\prime}$ for path $p_{a^\prime}=p^{xw}_{a^\prime}p^{wz}$ and capacity $c^{xw}_b$ (using Algorithm \ref{alg:best_schedule_edge_overlap_path}) and then add the action $a^\prime_i = (p_{a^\prime}, schedule_{a^\prime})$ to $A_{cand}$ (Algorithm \ref{alg:solve_brp} line \ref{algline:case_1_end}).
Now, because in both $a^\prime_i$ and $b_i$, the same path capacity $c^{xw}_b$ is used from $x$ to $w$, the same path from $w$ to $z$ is used, and $\tau^{xw}_{a^\prime} \leq \tau^{xw}_b$ (i.e., evacuees will reach the transit node $w$ earlier in $a^\prime_i$ than $b_i$), so: $$u_i(a^\prime_i, a_{-i}) \geq u_i(b_i, a_{-i})$$

So, in both cases an action $a^\prime_i$ will be added to $A_{cand}$ where:
\begin{align}
    \label{eq:case_1_candidate}
    u_i(a^\prime_i, a_{-i}) \geq u_i(b_i, a_{-i})
\end{align}
Now, because $a_i$ is the action with maximum utility over all actions in $A_{cand}$ (Algorithm \ref{alg:solve_brp} line \ref{algline:highest_utility_start}), so,
\begin{align}
    \label{eq:case_1_best}
    u_i(a_i, a_{-i}) \geq u_i(a^\prime_i, a_{-i})
\end{align}
Combining inequalities \ref{eq:case_1_candidate}, \ref{eq:case_1_best} we have: $u_i(a_i, a_{-i}) \geq u_i(b_i, a_{-i})$.
\end{proof}

It follows directly from Lemma \ref{lemma:brsa_correctness} that \bestresponsealgo{} returns a best response action.

\subsection{Time Complxity of \bestresponsealgo{}}
\subsubsection{Time complexity of \textit{backward\_update}}
Here, in worst case we will traverse all edges in the forest $\cF$ (constructed from the paths in $a_{-i}$). Because $\cF$ is a forest, the number of edges in $\cF$ is $O(n)$ ($n$ is number of nodes in the evacuation network $\cG$). Therefore, the time complexity of \textit{backward\_update} is $O(n)$.

\subsubsection{Time Complexity of \textit{preprocess}}
The length of the path of player $j$ is $O(n)$ (as the paths are simple path). Therefore, the time complexity of the loop in lines \ref{algline:path_time_loop_start}--\ref{algline:path_time_loop_end} is $O(n)$. Considering the hash table operations to be $O(1)$, time complexity of the loop in lines \ref{algline:update_loop_start}--\ref{algline:update_loop_end} is $O(n)$. 

Let $M_j$ denote the number of evacuees under player $j$. Then, the schedule of player $j$ can have at most $M_j$ elements in it (see Section \ref{sec:schedule_data_structure}). So, time complexity of processing player $j$ is: $O(M_jn)$. Time complexity of processing all players except $i$ is $O(n\sum_{j\in [N], j\neq i}M_j)$ or $O(Mn)$ ($M = \sum_{j\in [N]}M_j$, i.e. total number of evacuees).

For constructing the forest we will go over each player's route of length $O(n)$. Number of players is $N <= n$, so $O(n)$. So, complexity of constructing the forest is $O(n^2)$.

So, the total complexity of \textit{preprocess} is $O(n^2 + Mn)$.

\subsubsection{Time Complexity of Calculating Best Schedule (Algorithm \ref{alg:best_schedule_edge_disjoint_path}, \ref{alg:best_schedule_edge_overlap_path})}
In Algorithm \ref{alg:best_schedule_edge_disjoint_path}, because the path $p$ is disjoint, we can send at least one player at each timestep. So, the loop in line \ref{algline:disjoint_best_schedule_loop} can run at most $O(M_i)$ or $O(M)$ times. All operations within the loop are $O(1)$. So, running time of Algorithm \ref{alg:best_schedule_edge_disjoint_path} is $O(M)$.

In Algorithm \ref{alg:best_schedule_edge_overlap_path}, evacuees from player $j$ may be blocked by evacuees from other players. However, this blocking can happen in at most $\sum_{j \in [N], j \neq i}M_j$ timesteps. This is because, a single evacuee (under some player $j \neq i$) can cause a blocking for player $i$ in at most one timestep (because the paths are confluent, Assumption \ref{assumption:existing_routes_confluent}, and there is no waiting allowed on the roads). Together with the $\sum_{j \in [N], j \neq i}M_j$ blocking timesteps, player $i$ will also require at most $O(M_i)$ timesteps to send her evacuees. So, the loop in line \ref{algline:overlap_best_schedule_loop} will run for $O(\sum_{j \in [N]}M_j)$ or $O(M)$ times. Considering the hash table operations to be $O(1)$, the complexity of Algorithm \ref{alg:best_schedule_edge_overlap_path} is $O(M)$.

\subsubsection{Running Time of \bestresponsealgo{}}
Now, we calculate the time complexity of \bestresponsealgo{}. First, we perform preprocess with time complexity $O(n^2+Mn)$. Then, we will execute Algorithm \ref{alg:solve_brp}. In line \ref{algline:construct_forest_brp}, we construct the forest: $O(n^2)$. Lines \ref{algline:net_constr_start}--\ref{algline:net_constr_end} will take $O(m+n)$, here $n, m$ denote number of nodes, edges in $\cG$, respectively. Number of distinct edge capacity values is: $O(m)$. Constructing the network $\cG^\prime_c$ is: $O(m+n)$. Calculating the shortest paths in line \ref{algline:shortest_path} is: $O(m\log n)$. Number of safe nodes is: $O(n)$, calculating the best schedule on a disjoint path is: $O(M)$. So, complexity of lines \ref{algline:case_1_start}--\ref{algline:case_1_end} is $O(Mn)$. Similarly, the complexity of lines \ref{algline:case_2_start}--\ref{algline:case_2_end} is also $O(Mn)$. So, complexity of the entire loop starting at line \ref{algline:all_capacity_end} is $O(m(m+n+m\log n + Mn))$ or $O(Mmn + m^2\log n)$. As, number of actions in $A_{cand}$ is $O(n)$, running time of line \ref{algline:choose_best_response} is $O(n)$. So, total running time of calculating best response: $O(n^2 + Mn + Mmn + m^2\log n + n)$ or $O(n^2 + m^2\log n + Mmn)$.

\section{\seqactionalgo{} Details and Proofs}
\subsection{Proof of Theorem \ref{theorem:saa_correctness}}
\label{sec:before_tmax_proof}
\textbf{Theorem \ref{theorem:saa_correctness}:} \textit{\seqactionalgo{} returns a pure strategy Nash equilibrium (under the Confluent Route Constraint), where all players get a utility greater than $-M_1$}

\begin{proof}
We have already proved in Section \ref{sec:eq_calc} that \seqactionalgo{} returns a pure strategy Nash equilibrium under the Confluent Route Constraint. Here, we prove that with the outcome returned by \seqactionalgo{}, all evacuees reach a safe node within time $T_{max}$.

Let $seq$ denote the sequence of players used by \seqactionalgo{}, and $\hat{a}$ denote the outcome returned by \seqactionalgo{}.
Let, $d_i$ denote the number of evacuees under player $i$. For the first player $P_1$ in $seq$, no matter what route is chosen in $\hat{a}$, a possible departure time schedule for $P_1$ is as follows: starting from timestep $0$, one evacuee will leave at every timestep. This way all evacuees under $P_1$ will depart by timestep $d_1-1$. In $\hat{a}$, all evacuees under $P_1$ will also depart by timestep $d_1-1$ (otherwise $\hat{a}_1$ would not be $P_1$'s best action within \seqactionalgo{}). So, in $\hat{a}$, all evacuees under $P_1$ will reach a safe node within $(\tau + d_1 - 1)$, where $\tau = \sum_{e \in \cA}\tau_e$. Now for player $2$, no matter what route is chosen in $\hat{a}$, following is a possible departure time schedule: starting from timestep $(\tau + d_1)$, let one evacuee under $P_2$ leave at every timestep. Because the evacuees under $P_1$ will have already reached a safe node and there is no other evacuee in the network, all evacuees under $P_2$ will be able to depart within timestep $(\tau + d_1+d_2-1)$ and reach a safe node within time $(2\tau + d_1 + d_2 - 1)$. By the same argument used for $P_1$, in $\hat{a}$, all evacuees under $P_2$ will also depart by timestep $(\tau + d_1+d_2-1)$. Therefore, all evacuees under $P_1, P_2$ will reach safe node within timestep $(2\tau + d_1+d_2-1)$. Continuing this up to the last player $P_N$, in $\hat{a}$, all evacuees will reach a safe node within timestep $N\tau + d_1 + d_2 + ... + d_N -1 \leq n\tau + M - 1 = T_{max}$.
\end{proof}

\subsection{Assumption \ref{assumption:existing_routes_confluent} holds in \seqactionalgo{}}
Within \seqactionalgo{}, we use \bestresponsealgo{} for calculating best response. As, in \bestresponsealgo{}, we have Assumption \ref{assumption:existing_routes_confluent}, we have to ensure that it holds in every iteration of \seqactionalgo{}. 

Let $seq$ denote the sequence of players used by \seqactionalgo{}. For the first player $P_1$ in $seq$, Assumption \ref{assumption:existing_routes_confluent} holds as $a_{-i}$ is empty. Then, for the second player $P_2$ in $seq$ Assumption \ref{assumption:existing_routes_confluent} holds because $(i)$ $a_{-i}$ has only one route, $(ii)$ \bestresponsealgo{} makes sure that there is no capacity violation on $P_1$'s route and $(iii)$ we already proved that evacuees under $P_1$ will reach a safe node before $T_{max}$ (Section \ref{sec:before_tmax_proof}). Now, for the third player $P_3$, Assumption \ref{assumption:existing_routes_confluent} holds because: $(i)$ \bestresponsealgo{} ensures that the paths of $P_1, P_2$ are confluent, $(ii)$ there is no capacity violation on $P_2$'s path (and therefore cannot cause capacity violation on $P_1$'s path either), and $(iii)$ we already proved that evacuees under $P_1, P_2$ will reach safe node within $T_{max}$. Similarly, we can show that Assumption \ref{assumption:existing_routes_confluent} holds for all players upto the last player.

\subsection{Running Time of \seqactionalgo{}}
Within \seqactionalgo{}, we use \bestresponsealgo{} for every player. As the number of players $N \leq n$, so, the running time \seqactionalgo{} is $O(n(n^2 + m^2\log n + Mmn))$ or $O(n^3 + m^2n\log n + Mmn^2)$.

\bibliographystyle{IEEEtran}
\bibliography{references}

\end{document}